%% file: WorkingDraft.tex
\documentclass[12pt]{article}
\usepackage{tikz}
\usepackage{pgfplots}
\pgfplotsset{compat=1.18}

\input{macro}

\begin{document}
\include{math}
\include{notation}

\begin{titlepage}

\title{Screening for Choice Sets\thanks{Yingkai Li thanks the NUS Start-up Grant for financial support. We thank Joyee Deb, Piotr Dworczak, Nima Haghpanah, Marina Halac, Elliot Lipnowski, Harry Pei, Kai Hao Yang, and Jidong Zhou for helpful comments and suggestions.}}
\author{Tan Gan\thanks{Department of Managerial Economics and Strategy, London School of Economics. Email: \texttt{t.gan2@lse.ac.uk}}
\and Yingkai Li\thanks{Department of Economics, National University of Singapore.
Email: \texttt{yk.li@nus.edu.sg}}}

% \date{\today}
\date{}
\maketitle
\begin{abstract}
\noindent

We study a screening problem in which an agent privately knows which actions or technologies are feasible and can disclose only a subset to a principal. Once disclosed, feasible options are verifiable and their payoff consequences are publicly known, so private information concerns feasibility rather than payoffs, misreporting restricts the principal's choices directly rather than distorting her beliefs. Assuming feasible sets are ordered by inclusion, we establish a simple characterization of the optimal mechanism, where the principal either behaves as if there is no asymmetric information or locally provides no reward for better proposals. We derive comparative statics and illustrate the framework in applications to managing persuasion, action elicitation, and production-technology elicitation.

\end{abstract}


\JEL{D81, D82, D86}

\setcounter{page}{0}
\thispagestyle{empty}
\end{titlepage}
\pagebreak \newpage

\section{Introduction} \label{sec:introduction}
Screening and mechanism design are central pillars of modern economic theory, traditionally modeling private information as payoff-relevant heterogeneity over a commonly known set of feasible actions. In such Myersonian environments, agents possess hidden information about preferences, costs, or states, and incentive constraints arise because different types value the same allocation differently. Yet in many economic situations, the principal understands the payoff consequences of feasible actions but does not know which actions are feasible, relying instead on the agent to identify or disclose them. Engineers propose designs/algorithms, managers propose business strategies, firms propose experimental tests to regulators, and civil servants propose policy reforms. In each case, proposed options are verifiable once revealed, but feasible alternatives can be concealed beforehand. Misreporting therefore affects the principal's opportunity set rather than her beliefs about payoffs, a strategic tension that standard screening models are not designed to capture. This paper develops a screening framework in which private information concerns feasibility sets rather than payoff parameters, and studies how a principal should commit to decision rules that discipline strategic disclosure of feasible choice sets.

We formulate a screening model in which the agent privately observes a set of feasible technologies and may report only a subset of these technologies to the principal. The principal always has access to a default set of technologies, which induces a default choice set of payoff vectors—pairs $(u,v)$ representing the agent's and principal's payoffs—from which she can choose even if the agent reports no additional technologies. Each reported technology set expands this baseline and induces a larger choice set of payoff vectors from which the principal selects upon receiving the report. Unlike classical Myersonian screening models, the agent's private information here does not concern payoff-relevant parameters but the set of feasible technologies itself. This distinction has two fundamental consequences. First, by concealing feasible technologies, the agent's report directly alters the principal's choice set and hence her attainable payoff, even when the principal perfectly knows the agent's type. Second, once a feasible option is disclosed, its payoff consequences are publicly known, so no inference problem remains. In particular, when an agent with a larger feasible technology set mimics an agent with a smaller feasible technology set, he obtains exactly the smaller-set agent's payoff rather than an information rent. The principal's design problem is therefore to commit ex ante to a decision rule mapping reported technology sets to payoff choices so as to maximize her expected payoff while inducing truthful revelation of the agent's full feasible set.

Modeling private information as knowledge of feasible choice sets rather than payoff parameters naturally entails screening over high-dimensional objects. Without additional structure, this creates severe tractability challenges. A key modeling assumption is that technology sets are ordered by inclusion: more capable agents have access to larger technology sets. This nested structure preserves tractability while capturing a wide range of environments, as we illustrate in the applications. Methodologically, our approach aligns with the recent trend in multidimensional screening such as \cite{Yang2025}, which imposes a complete order on the type space while allowing rich heterogeneity in other dimensions and permitting the principal to employ high-dimensional screening instruments.

After presenting the model, \cref{subsec:applications} discusses three applications to illustrate how the abstract notions of technologies and choice sets correspond to concrete economic environments. In problems of managing persuasion, a sender privately knows which experiments or information structures can be conducted, while the receiver can verify any disclosed experiment but relies on the sender to make feasible experiments available. In problems of action elicitation, an expert privately knows which detailed actions or reforms can be implemented, while the principal cannot identify or implement these actions without the expert proposing them. In problems of production-technology elicitation, a manager privately knows which production methods, projects, or business strategies are feasible, while the principal can evaluate and authorize any disclosed strategy but cannot access undisclosed ones. In each case, feasible options are verifiable once revealed, concealment restricts the principal's feasible choice set, and more capable agents naturally possess larger sets of feasible options. We will apply our main result to solve these examples later.

We begin our formal analysis by isolating the role of private access to technologies independently of private information. Even when the principal perfectly knows the agent's technology set, the agent can still conceal feasible technologies and thereby restrict the principal's choice set. This feature has no analogue in standard screening models, where under complete information, the principal can directly select her preferred allocation. We show that disciplining such concealment requires extreme off-path incentives: the optimal complete-information mechanism takes a ``shoot-the-agent'' form in which any report that withholds feasible technologies is punished by assigning the agent the lowest payoff attainable under the default set of technologies. As a result, private access to technologies imposes a single restriction on the principal: the agent must be guaranteed a minimum utility determined by the worst punishment under the default set of technologies. 

To organize the analysis that follows, we summarize the complete-information benchmark by the complete information curve, defined as the agent's equilibrium payoff in the optimal mechanism under the complete information benchmark. For each technology set, this payoff equals the agent's utility under the principal's preferred feasible outcome whenever that utility exceeds the punishment guaranteed by the default set of technologies, and it equals the punishment level otherwise. This curve is always pointwise attainable, but it need not be monotone in set inclusion.

We then introduce private information about technology sets. We show that the principal's mechanism design problem can be decomposed into two steps. First, the principal chooses a promised utility function, assigning each reported technology set a payoff for the agent. Second, given any such promised utility, the principal selects her payoff-maximizing point in the corresponding feasible choice set that delivers this promise. While the second step may be computationally non-trivial in some applications, it is conceptually straightforward. This decomposition reduces the screening problem to choosing a promised utility function subject to incentive compatibility and technological feasibility.

A key observation is that, under nested technology sets, incentive compatibility takes a particularly simple form: the promised utility function must be weakly increasing in set inclusion and must guarantee the agent at least his default-report payoff. Moreover,  \cref{lem:range_of_promised_util} shows that under nested technology sets it is without loss of generality to restrict attention to promised utility functions that lie within the monotone envelope of the complete information curve—that is, between its lower and upper monotone envelopes. This is illustrated in \cref{fig:monotone-closures}. Crucially, once attention is restricted to this envelope, technological feasibility ceases to pose additional constraints: every weakly increasing promised utility function within the envelope is feasible. As a consequence, the screening problem collapses to choosing a one-dimensional monotone promised utility function within an explicitly constructed envelope determined by the complete-information benchmark.

\begin{figure}[htb]  % or [t] / [htbp] depending on placement preference

\centering

\begin{subfigure}[t]{0.49\textwidth}
    \centering
    \begin{tikzpicture}[scale=0.7]
        \begin{axis}[
            width=12cm, height=6cm,
            axis lines=middle,
            xmin=0, xmax=13,
            ymin=0, ymax=4.5,
            legend pos=north west,
            legend style={font=\Large},
            ticks=none
        ]
        % The Function f(x) / complete information
        \addplot [blue, thick, smooth] coordinates {
            (0.00,0.00) (0.32,0.41) (0.63,0.78) (0.95,1.10) (1.26,1.33) (1.58,1.47) 
            (1.89,1.52) (2.21,1.47) (2.53,1.34) (2.84,1.15) (3.22,0.89) (3.54,0.68) 
            (3.85,0.50) (4.17,0.40) (4.48,0.37) (4.80,0.44) (5.11,0.61) (5.43,0.88) 
            (5.75,1.21) (6.06,1.60) (6.44,2.09) (6.76,2.48) (7.07,2.83) (7.39,3.11) 
            (7.70,3.30) (8.02,3.39) (8.34,3.39) (8.65,3.29) (8.97,3.13) (9.28,2.93) 
            (9.66,2.66) (9.98,2.47) (10.29,2.32) (10.61,2.26) (10.92,2.28) (11.24,2.40) 
            (11.56,2.62) (11.87,2.92) (12.19,3.29) (12.57,3.77)
        };
        \addlegendentry{$\completeInfo(\type)$}
        
        % Upper monotone closure
        \addplot [red, dashed, very thick] coordinates {
            (0.00,0.00) (0.32,0.41) (0.63,0.78) (0.95,1.10) (1.26,1.33) (1.58,1.47) 
            (1.89,1.52) (2.21,1.52) (2.53,1.52) (2.84,1.52) (3.22,1.52) (3.54,1.52) 
            (3.85,1.52) (4.17,1.52) (4.48,1.52) (4.80,1.52) (5.11,1.52) (5.43,1.52) 
            (5.98,1.52) (6.06,1.60) (6.44,2.09) (6.76,2.48) (7.07,2.83) (7.39,3.11) 
            (7.70,3.30) (8.02,3.39) (8.34,3.40) (8.65,3.40) (8.97,3.40) (9.28,3.40) 
            (9.66,3.40) (9.98,3.40) (10.29,3.40) (10.61,3.40) (10.92,3.40) (11.24,3.40) 
            (11.56,3.40) (11.87,3.40) (12.27,3.40) (12.57,3.77)
        };
        \addlegendentry{$\uppercurve(\type)$}

        % Lower monotone closure
        \addplot [green!60!black, dashed, very thick] coordinates {
            (0.00,0.00) (0.32,0.37) (0.63,0.37) (0.95,0.37) (1.26,0.37) (1.58,0.37) 
            (1.89,0.37) (2.21,0.37) (2.53,0.37) (2.84,0.37) (3.22,0.37) (3.54,0.37) 
            (3.85,0.37) (4.17,0.37) (4.48,0.37) (4.80,0.44) (5.11,0.61) (5.43,0.88) 
            (5.75,1.21) (6.06,1.60) (6.44,2.09) (6.6,2.25) (7.07,2.25) (7.39,2.25) 
            (7.70,2.25) (8.02,2.25) (8.34,2.25) (8.65,2.25) (8.97,2.25) (9.28,2.25) 
            (9.66,2.25) (9.98,2.25) (10.29,2.25) (10.61,2.25) (10.92,2.28) (11.24,2.40) 
            (11.56,2.62) (11.87,2.92) (12.19,3.29) (12.57,3.77)
        };
        \addlegendentry{$\lowercurve(\type)$}
        \end{axis}
    \end{tikzpicture}
    \caption{Non-monotone function $\completeInfo(\type)$ with its upper monotone closure $\uppercurve(\type)$ and lower monotone closure $\lowercurve(\type)$.}
    \label{fig:monotone-closures}
\end{subfigure}
\hfill   % or ~ or \quad depending on desired spacing
\begin{subfigure}[t]{0.49\textwidth}
    \centering
    \begin{tikzpicture}[scale=0.7]
        \begin{axis}[
            width=12cm, height=6cm,
            axis lines=middle,
            xmin=0, xmax=13,
            ymin=0, ymax=4.5,
            legend pos=north west,
            legend style={font=\Large},
            ticks=none
        ]
        % Complete information curve
        \addplot [blue, thick, smooth] coordinates {
            (0.00,0.00) (0.32,0.41) (0.63,0.78) (0.95,1.10) (1.26,1.33) (1.58,1.47) 
            (1.89,1.52) (2.21,1.47) (2.53,1.34) (2.84,1.15) (3.22,0.89) (3.54,0.68) 
            (3.85,0.50) (4.17,0.40) (4.48,0.37) (4.80,0.44) (5.11,0.61) (5.43,0.88) 
            (5.75,1.21) (6.06,1.60) (6.44,2.09) (6.76,2.48) (7.07,2.83) (7.39,3.11) 
            (7.70,3.30) (8.02,3.39) (8.34,3.39) (8.65,3.29) (8.97,3.13) (9.28,2.93) 
            (9.66,2.66) (9.98,2.47) (10.29,2.32) (10.61,2.26) (10.92,2.28) (11.24,2.40) 
            (11.56,2.62) (11.87,2.92) (12.19,3.29) (12.57,3.77)
        };
        \addlegendentry{$\completeInfo(\type)$}

        % Monotone promise
        \addplot [purple, very thick] coordinates {
            (0.00,0.00) (0.32,0.41) (0.63,0.78) (0.95,1.10) (1.26,1.10)
            (1.26,1.10) (1.58,1.10) (1.89,1.10) (2.21,1.10) (2.53,1.10) (2.84,1.10) (3.22,1.10) (3.54,1.10)
            (3.85,1.10) (4.17,1.10) (4.48,1.10) (4.80,1.10) (5.11,1.10) (5.43,1.10) (5.65,1.10)
            (5.65,1.10) (6.06,1.60) (6.44,2.09) (6.76,2.48) (7.07,2.83) (7.39,3.11)
            (7.39,3.11) (7.70,3.11) (8.02,3.11) (8.34,3.11) (8.65,3.11) (8.97,3.11) (9.28,3.11)
            (9.66,3.11) (9.98,3.11) (10.29,3.11) (10.61,3.11) (10.92,3.11) (11.24,3.11)
            (11.56,3.11) (11.87,3.11) (12.04,3.11)
            (12.04,3.11) (12.57,3.77)
        };
        \addlegendentry{$\promise(\type)$}
        \end{axis}
    \end{tikzpicture}
    \caption{Optimal promised utility $\promise(\type)$ given the complete information curve $\completeInfo(\type)$.}
    \label{fig:promised-utility}
\end{subfigure}

\caption{Illustration of monotone closures and optimal promised utility.}
\label{fig:full_info_bound}

\end{figure}
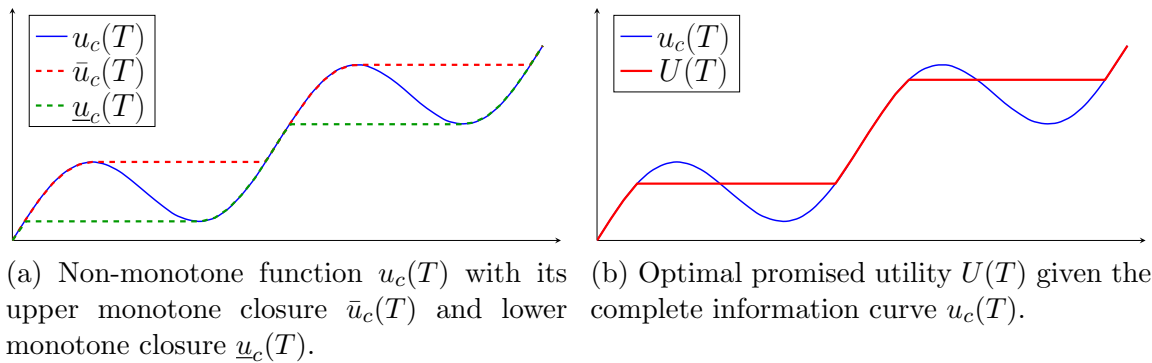

We further characterize the structure of the optimal mechanism. The principal's optimal policy exhibits a bang–bang structure. It partitions the type space into a finite collection of intervals. On each interval, either the principal implements the complete-information allocation, so each type receives exactly its complete-information payoff; or the principal keeps the agent's promised utility constant throughout the interval. This is illustrated in \cref{fig:promised-utility}.
Moreover, the number of flat segments is bounded by $K$, the number of decreasing segments of the complete-information curve. 
As a result, computing the optimal mechanism reduces to choosing at most $K$ pooling intervals and associated promise utilities that depart from the complete-information curve; in discrete cases, this computation can be implemented efficiently using dynamic programming.
The resulting bang–bang structure resembles institutional environments in which principals differentiate rewards finely across disclosed capabilities in some ranges but deliberately compress rewards in other ranges.

We then derive general comparative statics. The principal's equilibrium payoff increases when the distribution of agent capabilities improves in the sense of first-order stochastic dominance under set inclusion, and also when the principal's default set of technologies expands. 
The agent, by contrast, can be made weakly worse off---there exists an optimal selection of promised utilities under which his payoff weakly falls for every type---when the principal's default set of technologies expands, while the effect of improved capabilities is ambiguous. 
These results characterize when richer feasible technology sets benefit the principal and explain why enhanced agent capabilities may either help or hurt agents, depending on how they reshape the optimal bunching regions.

Finally, we apply the framework to three economically relevant environments: managing Bayesian persuasion when senders control experiment design, eliciting expert policy actions from informed civil servants, and contracting with CEOs who privately control additional business strategies under moral hazard. In each application, we show how the complete-information benchmark, its monotone envelope, and the resulting bang–bang structure deliver sharp and sometimes counterintuitive predictions—for example, why richer feasible technology sets may fail to raise agents' equilibrium payoffs, and how even a completely flat promised utility curve may generate rich economic predictions due to the nature of the optimization problem that we conceptually neglect in the general characterization.

\subsection{Literature} \label{sec:literature}

This paper relates to the literature on project selection, where an agent privately observes a set of feasible projects and may propose one to a principal who can verify the proposed project's characteristics and can reject or implement the project. \citet{ArmstrongVickers2010} provide the canonical model: the principal chooses a delegation or acceptance set, and the agent selects a feasible project within it. This captures one of our key frictions, as concealment restricts the principal's feasible set. However, for tractability, their paper assumes that an agent can report only one project and that the principal's policy must be deterministic—either accepting or rejecting. \citet{GuoShmaya2023} adopt a regret-minimization objective that helps them overcome some tractability issues and discuss the implications of allowing the agent to propose multiple projects simultaneously.

Our paper provides a unified yet tractable framework that embeds project proposals as a special case. In our model, the agent can report a subset of technologies, which impose structures over payoff pairs, instead of a single payoff pair or any arbitrary subset of payoff pairs. The principal can commit to taking any actions enabled by the reported technologies rather than simply accepting or rejecting a proposal. These modifications substantially expand the scope of the general model and enable us to study strategic proposals regarding signal structures, production technologies, and unknown actions.

This paper also relates to our earlier work \citet{BGL25}, which studies the specific example of managing persuasion, where the principal aims to discipline the strategic persuasion of an agent with private access to experiments/signals. To avoid tractability challenges, \citet{BGL25} adopts a robustness approach in which the principal minimizes worst-case regret over uncertainty in both the feasibility signal and the sender's partially aligned preferences. The present paper uses a classical Bayesian objective and focuses only on uncertainty over feasibility, thereby isolating the novel effect of screening on feasibility and providing a clearer contrast with the classical literature. 
\citet{CurelloSinander2025} analyze a dynamic disclosure problem in which the agent privately observes the arrival of a breakthrough that expands the feasible frontier. Since uncertainty concerns the timing of the breakthrough, their optimal mechanisms take the form of deadline mechanisms. Their characterization applies to a static variant in which the agent privately knows the arrival time from the outset, and this static variant can be embedded in our framework. An agent with an earlier arrival time has a larger choice set; the corresponding complete-information curve is strictly decreasing; and the optimal promised-utility curve is flat, so the optimal mechanism exhibits the same deadline structures.

Another relevant strand is mechanism design with verifiable evidence or partial disclosure, where agents may possess hard information that cannot be fabricated but can be withheld. Classic and modern contributions study how disclosure constraints affect implementability and mechanism structure (e.g., \citealp{GreenLaffont1986,BullWatson2007,DeneckereSeverinov2008,BenPorathDekelLipman2019}). Our model shares one similarity in that the ``message'' the agent can send depends on his type. 
However, in evidence models, the message is still ``cheap'' in the sense that it does not directly change payoffs. Upon observing the message, the principal must infer payoffs from the message structure, and this inference is often imperfect. In contrast, in our model, the agent can directly affect the principal's payoff by hiding technologies, and once technologies are reported (on or off path), the principal can always perfectly evaluate the consequences of her choices. 

Finally, the paper is related to the literature on multidimensional screening, which is generally intractable without strong structural assumptions. A recent and productive approach regains tractability by imposing an order on types. Our paper adopts this methodological idea—imposing an inclusion order on the feasible \emph{set} rather than on preferences. In our setting the principal chooses a payoff pair $(u,v)$, so the allocation is intrinsically two-dimensional even though types are ordered. As illustrated by recent excellent work such as \citet{Yang2025} and \citet{LoertscherMuir2025}, even when the agent's type is one-dimensional and satisfies monotonicity, the multidimensional nature of the allocation space makes the analysis nonstandard and nontrivial.  One theoretical contribution of our paper is to show that, under our specific assumptions, the problem admits a simple and tractable solution. That is not at all obvious a priori, even given a one-dimensional type space. That said, once the model is properly formulated, the techniques used to solve it are elementary relative to solving the general multidimensional screening problem; we discuss how they relate to and differ from ironing discussed in the literature. Thus, we view our main contribution as conceptual: the framework is novel, and it allows us to study a range of applications and derive interesting economic insights.

\section{Model} \label{sec:model}

\paragraph{Primitives}
We consider a principal-agent model of screening for choice sets. There is a compact set $\techs$ of all possible technologies $\tech$. For a set of technologies $\type \subseteq \techs$, we define 
\[C(\type): 2^\techs  \to   2^{\reals^2} \]
as the principal's choice set. That is, if the principal gets access to a set of technologies $T\subseteq \techs$, she will be able to choose a pair of payoffs $(u,v) \in C(T) \subseteq \reals^2$, where $(u,v)$ represents the payoffs to the agent and the principal, respectively. $C(T)$ is compact-valued and is monotone in set inclusion: $C(T') \subseteq C(T)$ for any $T'\subseteq T$.

The principal has a default set of technologies $\default \subset \techs$ from which she can choose, and she relies on the agent's report on additional technologies to expand her feasible choice set.
Each agent has private access to a subset of technologies $\type \subset \techs$, which is the agent's private information (type). Without loss of generality, we can always let $\default\subseteq \type$ so $C(\type)$ is the choice set that the principal faces if the type $\type$ agent truthfully reports all his technologies.
Let 
$$\types\equiv\{\type'\subseteq\techs: \default\subseteq\type'\}$$
be the type space of the agent. 

The principal does not observe the agent's private type $\type$ and holds a prior belief~$\dist$ with support $\types_{\dist}$. 
Due to the tractability challenge of high-dimensional screening problems, we primarily focus on the environment where the agent can be ranked by his capability: the more capable agent has a larger set of available technologies.
\begin{assumption}[Nested Technology Sets]
\label{assumption1}
The set of available technologies is nested. That is, for any $\type,\type'\in \types_{\dist}$, either $\type\subseteq \type'$ or $\type'\subseteq \type$.
\end{assumption}
As we will discuss through examples and analysis, this assumption is reasonable in many applications and creates tractability in high-dimensional screening problems while maintaining some degree of flexibility to capture the complex nature of choice sets.

\paragraph{Reports and Mechanisms}

An agent with type $\type$ can strategically report a subset $\type'$ such that $\default \subseteq \type' \subseteq \type$ to the principal. That is, he can conceal available technologies but can never falsify technologies that do not exist. This assumption is appropriate for technologies, as they are hard to invent but easy to verify. If the agent reports the default, then the principal can only choose from the default $\default$.

To discipline the strategic report of the agent, the principal can ex-ante commit to a decision rule $a(\type)$
\[  a(\type):  \types \to  \Delta C(\type) \]
that maps any reported set of technology, including reports outside the support of~$\dist$, to a distribution of payoff pairs within the feasible choice set induced by the technology. 
We assume both players are risk neutral and satisfy the von Neumann-Morgenstern expected utility representation. 
Therefore, since the principal can commit to arbitrary random actions, without loss of generality, we assume $a(\type)$ is a deterministic mapping to the expected payoff pairs such that
\[  a(\type):  \types \to   (u,v) \in \conv(C(\type)). \]
Throughout the rest of the paper, we assume that randomized actions are already included in the choice set. Hence, for every type $\type$, the set $C(\type)$ is convex and satisfies $C(\type)=\conv(C(\type))$.

Although the setting is different from the standard Myersonian screening models, the idea of the revelation principle still holds, and it is without loss of generality to focus on direct mechanisms where the agent truthfully reports his technology set $\type$. The objective of the mechanism design is:
\begin{align*}
    \max_{a(\cdot)} \quad& \expect[\dist]{a_2(\type)} \\
    \text{s.t.} \quad & a(\type)\in C(\type) \quad \forall \type\\
    & a_1(\type) \geq a_1(\type')  \quad \forall \type' \subset \type.
\end{align*}

\subsection{Applications}
\label{subsec:applications}
The language of our model is deliberately abstract and contains many layers to encompass various screening problems related to feasible choice sets. We start by discussing a few examples to illustrate how our general model relates to different applications. These applications are also stated in general terms. We will provide more detailed examples and characterize their solutions after the main result.

\paragraph{Managing Persuasion}
This problem concerns how a receiver with commitment power can discipline the strategic persuasion of a sender as in \citet{BGL25}. There is a state space $\Theta$ and an action space $A$. Both the sender's utility function $u(a,\theta)$ and the receiver's utility function $v(a,\theta)$ are publicly known. 
The receiver shares a common prior $G$ about the state $\theta$ with the sender.

The receiver's (principal's) default choice set is to take a (possibly random) action $a\in \Delta A$ for all states based on the prior $G$, so the choice set under default set of technologies is
\[ C(\default) = \{ (u,v) | ~ \exists a \in \Delta A, ~ \text{s.t.}~ u= \expect[\theta\sim G]{u(a,\theta)}, ~ v= \expect[\theta\sim G]{v(a,\theta)} \}. \]

The sender can conduct experiments to persuade the receiver, as in the literature on Bayesian persuasion \citep{kamenica2011bayesian}. An experiment consists of a signal space $S$ and a signal mapping $\sigma: \Theta \to S$. Given any experiment $(S,\sigma)$, the principal can adopt a (possibly random) strategy $a$ that maps from $S$ to $\Delta A$. A technology $t$ corresponds to an experiment $(S,\sigma)$. The choice set with type $\type$ is
\begin{align*}
    C(T) = \{ (u,v) | ~ &\exists (S,\sigma) \in T, ~ a: S \to \Delta A \,\\
    &~\text{s.t.} ~ u = \expect[\theta\sim G, s\sim\sigma(\theta)]{u(a(s),\theta)}, ~ v = \expect[\theta\sim G, s\sim\sigma(\theta)]{v(a(s),\theta)} \} .
\end{align*}

However, the agent may not have the ability to perfectly fine-tune the experiments. Thus, there is only a set of experiments available, which is captured by the agent's private type $\type$. The agent can strategically report a subset of available experiments. Anticipating this, the principal optimally commits to the choice of the experiment and her action upon the signal realization.

In this setting, \cref{assumption1} requires that the agent be ranked by how many experiments they can conduct: a more capable agent can conduct more experiments in terms of set inclusion. \cref{assumption1} does not impose any restrictions on what these experiments are.
Note that our general framework can also accommodate the cost of the experiment $c(\sigma)$, which may vary with the choice of the experiment. This simply shifts the technology $t$ downward by the cost of the experiment on the coordinate axis that represents the payoff of the agent.

\paragraph{Action Elicitation}
This problem concerns how a principal, who has a good understanding of the state but lacks the expertise or time to identify what detailed strategies are feasible, can elicit feasible actions from the agent.
There is a state space $\theta\in\Theta$ drawn from common prior $G$, and a signal structure $\sigma:\Theta \to S$ that is privately observed by the principal. An action $a: \Theta \to (u,v)$ is a mapping from states to payoff pairs. The principal has a known set of actions $A_0$ from which she can choose. Thus, her default choice set is
\begin{align*}
C(\default) = \{ (u,v) | ~ &\exists a: S \to \Delta A_0, ~ \\
&\text{s.t.}~ u= \expect[\theta\sim G,s\sim\sigma(\theta)]{u(a(s),\theta)}, ~ v= \expect[\theta\sim G,s\sim\sigma(\theta)]{v(a(s),\theta)} \}.
\end{align*}
The agent has private access to some additional actions, such as a detailed feasible schedule or a novel algorithm. He can choose to conceal information in his reports, but he cannot falsify an action, as the principal will be able to easily verify and understand the payoff consequences.
A technology $t$ corresponds to an additional action. The choice set with type $\type$ is
\begin{align*}
    C(T) = \{ (u,v) | ~ &\exists  a: S \to \Delta (A_0\cup T) , \\
    &~\text{s.t.} ~ u = \expect[\theta\sim G, s\sim\sigma(\theta)]{u(a(s),\theta)}, ~ v = \expect[\theta\sim G, s\sim\sigma(\theta)]{v(a(s),\theta)} \} .
\end{align*}
In this setting, \cref{assumption1} requires that the agent is ranked by how many additional actions they have private access to: a more capable agent has more actions in terms of set inclusion.

\paragraph{Eliciting Production Technologies}
In this example, we highlight that the subsequent interaction between the principal and the agent, after the agent reports his set of available technologies, can be sophisticated.
Consider the application where the board of a company (the principal) is contracting with a CEO (the agent). A business strategy $(G,c)$ consists of a distribution $G\in \Delta \reals$ over the company's profit and the cost of the agent. The company has an initial set of business strategies $\mathcal{G}_0$, from which the agent can privately choose. The principal can use a non-negative wage scheme $w(\cdot): \reals \to \reals^+$ to provide incentives for the agent. Thus, the principal's default choice set is
\begin{align*}
    C(\default) &= \{ (u,v) | ~ \exists w(\cdot), ~ (G,c) \in \mathcal{G}_0 , ~\text{s.t.} ~ v = \int y-w(y) \dd G(y),\\ & \qquad u = \int w(y)-c \dd G(y) \geq \max_{(G',c')\in \mathcal{G}_0} \int w(y)-c' \dd G'(y) \} .
\end{align*}

The agent may have some other business strategies, but unlocking them requires the principal's permission or even additional investment. A technology $\tech$ corresponds to a set of business strategies $\mathcal{G}_t$, which may not necessarily be a singleton. For example, it is possible that a technology simultaneously introduces a productive business strategy and a shirking business strategy, and they cannot be separated.

After the agent reports the available set of additional business strategies (technology) $\cup_{\tech \in \type} \mathcal{G}_\tech $,  the principal can permit a subset $\type'$ of them: $\mathcal{G}(\type') = \cup_{\tech \in \type' } \mathcal{G}_\tech  $. With such selective permission, the set of payoffs that she can induce via a wage scheme is
\begin{align*}
    \tilde{C}(\mathcal{G}(\type') ) &= \{ (u,v) | ~ \exists w(\cdot), ~ (G,c) \in \mathcal{G}_0 \cup \mathcal{G}(\type') , ~\text{s.t.} ~ v = \int y-w(y) \dd G(y),\\ & \qquad u = \int w(y)-c \dd G(y) \geq \max_{(G',c')\in \mathcal{G}_0 \cup \mathcal{G}(\type')} \int w(y)-c' \dd G'(y) \} .
\end{align*}

Note that due to the complex nature of incentive compatibility, $\tilde{C}(\mathcal{G}(\type') )$ may not be monotonic in set inclusion in $\type'$. However, because the principal has the freedom to choose what $\mathcal{G}'$ to permit and what wage scheme to provide. Her ultimate choice set is monotonic in set inclusion in $T$:
\begin{gather*}
    C(T) = \cup_{\type' \subseteq \type}    \tilde{C}(\mathcal{G}(\type') ).
\end{gather*}
In this example, \cref{assumption1} holds as long as a more capable agent has more available business plans.

\section{Analysis}
\subsection{Complete Information Benchmark}
The agent in our model not only has private information about what technology is available but also has private access to these technologies. To highlight the difference between these two elements, we start with the complete information benchmark where the principal's prior $\dist$ has a degenerate support $\types_{\dist}=\{ \type^* \}$ for a fixed and commonly known type $\type^*$.

Ideally, the principal wants the agent to report $\type^*$, after which the principal can select her favorite choice. However, because the agent has private access to the technologies, he can choose to report only a subset $\type\subseteq \type^*$, which limits the principal's choice set. This highlights the first difference between our framework and the standard Myersonian framework: even in the complete information benchmark, the agent can directly affect the principal's payoff by changing the principal's available options. 

To discipline the agent's strategic selection of technologies, it is without loss of generality for the principal to use the ``shoot-the-agent'' mechanism. That is, the principal commits to minimizing the agent's payoff whenever the agent reports $T\neq \type^*$. 
The value of this minimum payoff is clearly decreasing in $T$ in set inclusion. Thus, among all $T \neq \type^*$, the best choice for the agent is to report the default: $T= \default$. The corresponding value $\utilmin$, where
\begin{equation}
\label{eq_outsideoption}
     \utilmin = \min_{(u,v) \in C(\default)} u,
\end{equation}
is the minimum value that the principal has to offer to the agent on the equilibrium path to ensure that he reports $T$ truthfully. Therefore, we formally define the ``shoot-the-agent'' mechanism as
\begin{align*}
    a(T) \in  \begin{cases}
        \argmin_{(u,v) \in C(\type)} u    \quad & \text{ if } T \neq \type^* \\
        \argmax_{(u,v) \in C(\type), u \geq \utilmin } v   \quad & \text{ if } T =\type^*.
    \end{cases}
\end{align*}

\begin{proposition}
    \label{prop_completeinformation}
    In the complete information benchmark, ``shoot-the-agent'' is an optimal mechanism.
\end{proposition}
For the purpose of subsequent analysis, we are interested in understanding the agent's payoffs under the optimal mechanism in the complete information benchmark. 
To simplify the exposition, we impose a generic assumption that the optimal choice of the principal with any technology set $\type\in \types$ is unique. All our results generalize without this assumption and we delay the discussion in the extension.
\begin{assumption}
\label{assumption_generic}
The optimal choice of the principal 
$\principalOpt(T)=\argmax_{(u,v) \in C(\type)} v$ is unique given any type $\type$.
\end{assumption}
With this assumption, $\principalOpt_1(T)$ is the agent's payoff when the principal chooses her optimal choice within $T$, and we know
\begin{corollary}\label{cor:agent_payoff_complete}
    Under the optimal mechanism of the complete information benchmark, the agent's payoff is
    \begin{equation}
    \label{eq_completeinformation}
        \completeInfo(T) =   \max \{ \utilmin, \principalOpt_1(T)    \}.
    \end{equation}
\end{corollary}
We denote $\completeInfo(T)$ as the complete information curve.
In the absence of private information about the technologies $T$, private access to the technology gives the agent very minimal power: he is willing to reveal all he has and let the principal choose her preferred option as long as it is weakly better than the worst punishment under the default set of technologies.

\subsection{Optimal Mechanism}
After discussing the complete information benchmark and introducing the complete information curve, we are ready to discuss the optimal mechanism with private information.

\paragraph{Promised Utility Function}

The first step of the analysis is to point out that the design problem can be separated into two independent problems. Unless otherwise specified, we always impose \cref{assumption1}. Therefore, to simplify notation, we also denote the agent's type $\type$ as a real number in $[0,1]$, where a real number $\type\leq \type'$ implies the technology set $\type \subseteq \type'$.

For an (incentive-compatible) direct mechanism $a(\type)$,  $\promise(\type)=a_1(\type)$ is the agent's payoff when he reports $\type$. We call this $\promise(\type)$ function the promised utility function. Once the promised utility function is given, the rest of the design problem is non-strategic:
\begin{lemma}
\label{lemma_PromiseUtility}
      It is without loss of optimality to focus on direct mechanisms $a(\type)$ where 
\begin{align*}
     a(\type) = \argmax_{(u,v)}  &\quad v, \\
     \text{\rm s.t.}  \quad&\quad (u,v) \in C(\type), u=\promise(\type)
\end{align*}
\end{lemma}
In some applications, e.g., business strategies with moral hazard, the analysis of the exact boundary of $C(\type)$ can be challenging.
Nevertheless, it is conceptually simple, and we will define
\begin{gather*}
    \optv(\type,u) = \max_{ (u',v') \in C(\type), u'=u}   v',
\end{gather*}
and proceed as if we fully understand $V(T,u)$. This decomposition helps us to smooth out the complexity of the problem and identify general properties independent of the detailed structure of the feasible choice set $C(\type)$. 
%Note that $V(T,u)$ is a concave function in $u$ by definition.

Not all promised utility functions can be achieved by some choice $(u,v)\in C(T)$ as the principal is limited by the choice set. The second step is to characterize which promised utility functions are (technologically) feasible.
\begin{lemma}
\label{lemma_feasible}
    $\promise(\type)$ is feasible if and only if $\promise(\type) \in [\lowerfeasible(\type),\upperfeasible(\type)  ]$, where
    \begin{align*}
        \lowerfeasible(\type) = \min_{(u,v)\in C(\type)} u, \quad \upperfeasible(\type) = \max_{(u,v)\in C(\type)} u.
    \end{align*}
    $\lowerfeasible(\type)$ is decreasing in $\type$ and $\upperfeasible(\type)$ is increasing in $\type$.
\end{lemma}
The feasibility constraint in optimal control problems is often quite messy. It introduces Lagrange (costate) multipliers that may or may not bind, obscuring the economic intuition of the solution and making it difficult to derive general predictions. However, as we will show in \cref{lem:range_of_promised_util}, under our set-inclusion assumption (\cref{assumption1}), the feasibility constraints do not bind in the optimal mechanism.

The third step is to characterize which promised utility functions $\promise(\type)$ are incentive compatible. Clearly, if $\type\geq \type'$, then the type $\type$ agent always has the option to report his subset $\type'$ and gets $\promise(\type')$. Thus, $\promise(\type)$ must be weakly increasing in $\type$. In fact, this is the only incentive constraint. 
\begin{lemma}
    \label{lemma_Monotone} A promised utility function $\promise(\type)$ is incentive compatible if and only if
\begin{enumerate}[ ]
  \item (i) $\promise(\type)$ is weakly increasing;
  \item (ii) $\promise(\default) \geq \utilmin$.
\end{enumerate}
\end{lemma}
\cref{lemma_Monotone}  highlights another difference between our screening model of feasibility and the Myersonian screening model of payoff-relevant information. In our model, high types only get a weakly higher payoff relative to low types, instead of obtaining a strictly higher payoff which depends on the mechanism's allocation rule. In fact, when high types misreport as a low type, they receive exactly the same payoff as the low type; whereas in the Myersonian screening model, high types obtain a strictly higher payoff even with the same allocation, which creates information rent.  This difference comes from the fact that the principal does not need to make any inferences once the set of technologies is reported: she can perfectly evaluate the consequences of all her choices.

With the help of \cref{lemma_PromiseUtility,lemma_feasible,lemma_Monotone}, 
we can rewrite the principal's design problem as 
\begin{align}
    \max_{\promise(\cdot)} &\quad \int  ~ \optv(\type,\promise(\type))  \dd F(\type) \label{eq_designproblem} \tag{OPT}\\
    \text{s.t.} &\quad \promise(\type) \text{ is weakly increasing}, \nonumber\\
        & \quad \promise(\type) \in [\lowerfeasible(\type),\upperfeasible(\type)  ], \nonumber\\
    &\quad \promise(\default) \geq \utilmin.  \nonumber
\end{align}

\paragraph{Optimal Promised Utility Functions}

Perhaps surprisingly, the complete-information curve $\completeInfo(T)$, defined in \Cref{eq_completeinformation}, is crucial to the characterization of the optimal promised utility function. To see this, we start with two simple observations. First, by construction, $\completeInfo(T)$ is always feasible. Second, if $\completeInfo(T)$ is increasing in $T$, then it is the optimal promised utility function, and the principal can attain the complete-information payoff. The central question, therefore, is how to characterize the optimal promised utility function when the complete-information curve is not monotone. To this end,
we  define the upper monotone closure $\uppercurve(\type)$ as the minimum monotone function such that 
$\uppercurve(\type) \geq \completeInfo(\type)$ for all $\type$. 
That is, 
\begin{align}
\uppercurve(\type) = \max_{\type'\leq \type} \completeInfo(\type'), \quad \forall \type \in \types.
\end{align}

See the red dashed line in \Cref{fig:full_info_bound} for an illustration. Similarly, the lower monotone closure $\lowercurve(\type)$ is the maximum monotone function such that 
$\lowercurve(\type) \leq \completeInfo(\type)$ for all $\type$, and
\begin{align}
\lowercurve(\type) = \min_{\type'\geq \type} \completeInfo(\type'), \quad \forall \type \in \types.
\end{align}

A promised utility function $U(T)$ is in the monotone envelope (of the complete information curve) if
\begin{align}
   \promise(\type) \in [\lowercurve(\type),\uppercurve(\type)], \quad \forall \type \in \types.
\end{align}
It turns out that it is without loss of generality to focus on a monotone promised utility function in the monotone envelopes. The intuition is that we can always project the promised utility function into the monotone envelope while keeping the monotonicity, as illustrated in \Cref{fig:lemma4_projection}, to improve the principal's payoff. Moreover, once we focus on the monotone promised utility function within the monotone envelope, the feasibility constraint can be neglected.

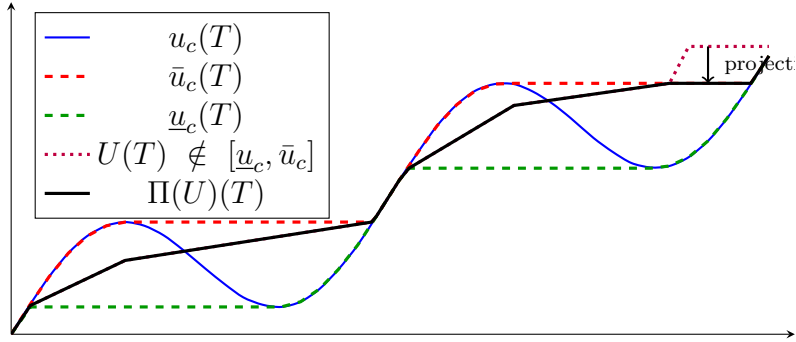
\begin{figure}
\centering
\begin{tikzpicture}
    \begin{axis}[
        width=12cm, height=6cm,
        axis lines=middle,
        xmin=0, xmax=13,
        ymin=0, ymax=4.5,
        legend pos=north west,
        ticks=none
    ]
    % The Function uc(T)
    \addplot [blue, thick, smooth] coordinates {
        (0.00,0.00) (0.32,0.41) (0.63,0.78) (0.95,1.10) (1.26,1.33) (1.58,1.47) 
        (1.89,1.52) (2.21,1.47) (2.53,1.34) (2.84,1.15) (3.22,0.89) (3.54,0.68) 
        (3.85,0.50) (4.17,0.40) (4.48,0.37) (4.80,0.44) (5.11,0.61) (5.43,0.88) 
        (5.75,1.21) (6.06,1.60) (6.44,2.09) (6.76,2.48) (7.07,2.83) (7.39,3.11) 
        (7.70,3.30) (8.02,3.39) (8.34,3.39) (8.65,3.29) (8.97,3.13) (9.28,2.93) 
        (9.66,2.66) (9.98,2.47) (10.29,2.32) (10.61,2.26) (10.92,2.28) (11.24,2.40) 
        (11.56,2.62) (11.87,2.92) (12.19,3.29) (12.57,3.77)
    };
    \addlegendentry{$\completeInfo(\type)$}
    
    % Upper Monotone Closure
    \addplot [red, dashed, very thick] coordinates {
        (0.00,0.00) (0.32,0.41) (0.63,0.78) (0.95,1.10) (1.26,1.33) (1.58,1.47) 
        (1.89,1.52) (2.21,1.52) (2.53,1.52) (2.84,1.52) (3.22,1.52) (3.54,1.52) 
        (3.85,1.52) (4.17,1.52) (4.48,1.52) (4.80,1.52) (5.11,1.52) (5.43,1.52) 
        (5.98,1.52) (6.06,1.60) (6.44,2.09) (6.76,2.48) (7.07,2.83) (7.39,3.11) 
        (7.70,3.30) (8.02,3.39) (8.34,3.40) (8.65,3.40) (8.97,3.40) (9.28,3.40) 
        (9.66,3.40) (9.98,3.40) (10.29,3.40) (10.61,3.40) (10.92,3.40) (11.24,3.40) 
        (11.56,3.40) (11.87,3.40) (12.27,3.40) (12.57,3.77)
    };
    \addlegendentry{$\uppercurve(\type)$}

    % Lower Monotone Closure
    \addplot [green!60!black, dashed, very thick] coordinates {
        (0.00,0.00) (0.32,0.37) (0.63,0.37) (0.95,0.37) (1.26,0.37) (1.58,0.37) 
        (1.89,0.37) (2.21,0.37) (2.53,0.37) (2.84,0.37) (3.22,0.37) (3.54,0.37) 
        (3.85,0.37) (4.17,0.37) (4.48,0.37) (4.80,0.44) (5.11,0.61) (5.43,0.88) 
        (5.75,1.21) (6.06,1.60) (6.44,2.09) (6.60,2.25) (7.07,2.25) (7.39,2.25) 
        (7.70,2.25) (8.02,2.25) (8.34,2.25) (8.65,2.25) (8.97,2.25) (9.28,2.25) 
        (9.66,2.25) (9.98,2.25) (10.29,2.25) (10.61,2.25) (10.92,2.28) (11.24,2.40) 
        (11.56,2.62) (11.87,2.92) (12.19,3.29) (12.57,3.77)
    };
    \addlegendentry{$\lowercurve(\type)$}

    % ------------------------------------------------------------
    % Promised utility U(T) (monotone) but NOT in the envelope
    % (It exceeds the upper envelope for high types.)
    \addplot [purple, dotted, very thick] coordinates {
        (0.00,0.00)
        (0.32,0.39)
        (1.89,1.00)
        (5.98,1.52)
        (6.06,1.60)
        (6.44,2.09)
        (6.60,2.25)
        (8.34,3.10)
        (10.92,3.40)
        (11.24,3.90)
        (12.57,3.90)
    };
    \addlegendentry{$U(\type)\ \notin\ [\lowercurve,\uppercurve]$}

    % Projection of U onto the monotone envelope (here: clamp to upper envelope)
    % Pi(U)(T) = min{ U(T), \bar u_c(T) }.
    \addplot [purple, very thick] coordinates {
        (0.00,0.00)
        (0.32,0.39)
        (1.89,1.00)
        (5.98,1.52)
        (6.06,1.60)
        (6.44,2.09)
        (6.60,2.25)
        (8.34,3.10)
        (10.92,3.40)
        (11.24,3.40)
        (12.27,3.40)
        (12.57,3.77)
    };
    \addlegendentry{$\Pi(U)(\type)$}

    % Improvement arrow / annotation (Lemma 4 idea)
    \draw[->, thick] (axis cs:11.56,3.90) -- (axis cs:11.56,3.40);
    \node[anchor=west] at (axis cs:11.65,3.65)
      {\scriptsize projection lowers promise $\Rightarrow$ improves principal payoff};

    \end{axis}
\end{tikzpicture}
\caption{Illustration for Lemma 4: a monotone promised utility $U(\type)$ that lies outside the monotone envelope can be projected into the envelope by clamping to $\uppercurve(\type)$ where it violates feasibility/optimality; this projection weakly increases the principal's payoff.}
\label{fig:lemma4_projection}
\end{figure}

\begin{lemma}\label{lem:range_of_promised_util}
Any promised utility function in the monotone envelope  is feasible. Moreover, there exists a monotone promised utility function in the monotone envelope, which solves the design problem \eqref{eq_designproblem}.
\end{lemma}

\cref{lem:range_of_promised_util} greatly simplifies the design problem: we just need to search over all monotone promised utility functions within the monotone envelope without worrying about any other constraints. It also implies that the design problem \eqref{eq_designproblem} admits a solution.

To prepare for a sharper characterization of the optimal promised utility, we quantify the non-monotonicity of the complete-information curve. Let 
$K$ denote the number of maximal intervals on which the complete-information curve is strictly decreasing. 
Formally,\footnote{We allow the decreasing interval $I$ to be degenerate, which corresponds to discontinuous downward jumps of $u_c(\cdot)$.}
\[
K \equiv 
\#\Big\{ I \subseteq [0,1] :\;
I \text{ is a maximal interval such that } 
u_c(\cdot) \text{ is strictly decreasing in } I 
\Big\}.
\]

The principal's optimal policy features a bang--bang structure across intervals: on each interval, she either implements the complete-information allocation or keeps the agent's promised utility constant throughout. Moreover, the number of constant-promise intervals is bounded by $K$.

\begin{figure}
\centering
\begin{tikzpicture}
    \begin{axis}[
        width=12cm, height=6cm,
        axis lines=middle,
        xmin=0, xmax=13,
        ymin=0, ymax=4.5,
        legend pos=north west,
        ticks=none
    ]

    % ----------------------------
    % Complete information curve (blue): smoother version
    \addplot [blue, thick, smooth] coordinates {
        (0.00,0.00)
        (0.50,0.03) (1.00,0.10) (1.50,0.22) (2.00,0.40) (2.50,0.62)
        (3.00,0.90) (3.50,1.20) (4.00,1.50) (4.50,1.9) (5.00,2.30)
        (5.50,2.7) (6.00,2.90) 
        (6.50,2.90) (7.0,2.77) (7.30,2.6) (7.50,2.40) (7.70,2.1) (8.0,1.8) (8.40,1.62) (8.60,1.55) (9.00,1.50)
        (9.00,1.50) (9.10,1.51) (9.20,1.54) (9.30,1.59) (9.40,1.66) (9.50,1.75) (9.60,1.86)
        (9.70,1.99) (9.80,2.14) (10.00,2.5) (10.50,3.30) 
        (11.00,3.70) (11.50,3.92) (12.00,4.08) (12.50,4.20) (13.00,4.3) 
    };
    \addlegendentry{$\completeInfo(\type)$}

    % ----------------------------
  
    % ----------------------------
    % U(\type): increasing, smooth-looking, strictly inside the envelope; intersects downward slope at (8.40,2.15)
    \addplot [purple, dotted, very thick, smooth] coordinates {
        (0.00,0.00) (0.50,0.03) (1.00,0.10) (1.50,0.22) (2.00,0.40) (2.50,0.62) (3.00,0.90) (3.50,1.20) (4.00,1.50)
        (4.00,1.50) (5.00,1.75) (6.00,1.95) (7.00,2.08) (7.60,2.15) (8.00,2.18) (9.00,2.3) (9.20,2.35) (9.40,2.45) (9.55,2.50) 
        (9.80,2.7) (10.00,2.8) (10.25,2.9)
        (10.25,2.9) (10.50,3.30) (11.00,3.70) (11.50,3.92) (12.00,4.08) (12.50,4.20) (13.00,4.3) 
    };
    \addlegendentry{$U(\type)$}

    % ----------------------------
    % \hat U(\type): constant within the envelope; intersects the downward slope at the same point as U
    \addplot [purple, very thick] coordinates {
        (0.00,0.00) (0.50,0.03) (1.00,0.10) (1.50,0.22) (2.00,0.40) (2.50,0.62) (3.00,0.90) (3.50,1.20) (4.00,1.50) (4.50,1.9) (4.80,2.15)
        (4.80,2.15) (8.00,2.15) (9.80,2.15)
        (9.80,2.15) (10.00,2.5) (10.50,3.30) 
        (11.00,3.70) (11.50,3.92) (12.00,4.08) (12.50,4.20) (13.00,4.3) 
    };
    \addlegendentry{$\hat U(\type)$}
    
    \draw[->, thick] (axis cs:5,1.8) -- (axis cs:5,2.1);
    \end{axis}
\end{tikzpicture}
\caption{Illustration of \cref{thm:optimal_promise}.}
\label{fig:thm1}
\end{figure}

\begin{theorem}[Optimal Mechanisms]
\label{thm:optimal_promise}
There exists a partition of the type space into intervals $\{I_j\}_{j \in J_1 \cup J_2}$ such that $|J_1|-1\leq |J_2|\leq K$ and
\begin{enumerate}
\item If $j\in J_1$, $\promise(\type) = \completeInfo(\type)$ for any type $\type\in I_j$;
\item If $j\in J_2$, $\promise(\type)$ is a constant function in $I_j$.
\end{enumerate} 
\end{theorem}

The intuition behind \cref{thm:optimal_promise} is transparent. For each type, the principal's payoff $\optv(\type,\cdot)$ is concave in the promised utility and is maximized exactly at the complete-information level $\completeInfo(\type)$, so absent incentive constraints she would set $\promise(\type)=\completeInfo(\type)$ type by type. The only binding constraint is that $\promise$ be weakly increasing (\cref{lemma_Monotone}). Wherever $\completeInfo$ is itself increasing, this constraint does not bite and the principal can deliver the complete-information allocation. The constraint bites only where $\completeInfo$ slopes downward: there, a monotone schedule cannot follow $\completeInfo$ down, so the principal must hold $\promise$ above $\completeInfo$ on a region to its left and below $\completeInfo$ on a region to its right. Because moving $\promise$ toward $\completeInfo$ always helps and moving it away always hurts (by the concavity of $\optv(\type,\cdot)$), the cheapest way to restore monotonicity is to flatten $\promise$ at a single level across the offending region rather than distort it elsewhere---which is exactly why the only departures from the benchmark are flat segments, and why each such segment sits on a downward-sloping piece of $\completeInfo$. This also explains the count: each strictly decreasing segment of $\completeInfo$ forces at most one flat region, giving the bound $|J_2|\le K$.

\cref{thm:optimal_promise} provides a rationale for two seemingly abnormal organizational behaviors: why institutions may deliberately limit incentives for marginally better proposals, and why they may fail to fully exploit newly disclosed feasible options. More strikingly, it shows that these two behaviors cannot appear in isolation in an optimal mechanism but instead may arise jointly. According to \cref{thm:optimal_promise}, either the principal fully exploits added feasibility—so richer feasible option sets are used in the best possible way for the principal—and continuous, strict incentives are provided, or she must both compress rewards and deliberately refrain from myopic exploitation to sustain truthful disclosure. One distortion cannot arise without the other.

This joint emergence of reward compression and deliberate constraints from the efficient exploitation of feasible options has clear organizational implications. In some capability ranges, better proposals lead to both better internal decisions and better treatment of the proposer: additional experiments, richer action plans, or more sophisticated business strategies are actively used to improve organizational outcomes, and the agent is rewarded accordingly. In other ranges, institutions deliberately maintain stable treatment even as proposals improve. Committees may follow standard approval rules, boards may adhere to fixed compensation packages, or regulators may apply uniform authorization pathways, despite being presented with richer feasible options. In these regions, additional disclosed capabilities expand the principal's internal choice set but are not fully exploited in the short run, because committing to stable treatment is necessary to sustain transparent disclosure of what is feasible.

\paragraph{Computational simplicity.}
\cref{thm:optimal_promise} reduces the entire infinite-dimensional design problem to choosing at most $K$ scalar cutoffs, namely the points at which the optimal promised utility departs from the complete-information curve. Thus, the optimal cutoffs can be located by a one-dimensional search over the $K$ decreasing segments of $\completeInfo$. In the discrete case, this can be implemented efficiently via dynamic programming. The characterization is thus not merely qualitative but delivers an easily computable description of the optimal mechanism.

\paragraph{Relation to ironing.}
The optimization \eqref{eq_designproblem} is mathematically reminiscent of ironing in mechanism design, and the interval partition in \cref{thm:optimal_promise} can be viewed as ironing applied to our setting. The difference is that \eqref{eq_designproblem} carries the additional feasibility constraints of \cref{lemma_feasible}, which \cref{lem:range_of_promised_util} shows to be non-binding. Despite this mathematical resemblance, the economics behind the optimization differ in two essential ways. First, the objective in \eqref{eq_designproblem} is the principal's \emph{actual} payoff $\optv(\type,\promise(\type))$, not a virtual-value function: because the principal can perfectly evaluate the consequences of every reported set, the agent earns no information rent, and there is no wedge between true and virtual surplus. Consequently, the optimal mechanism deviates from the complete-information benchmark (the first best) if and only if some promised-utility region is flattened---bunching occurs precisely when, and only when, the principal sacrifices myopic exploitation of feasible options. Second, the improvement underlying \cref{thm:optimal_promise} is distributionally robust: the same flattening operation raises the principal's payoff \emph{for every} type distribution $\dist$, since it improves the payoff pointwise in $\type$. This is not true of standard ironing, where the curvature being ironed depends on $\dist$ through the virtual value. We exploit this robustness in \cref{sec:mcs} below: it implies that the mechanisms identified in \cref{thm:optimal_promise} are exactly the undominated mechanisms in the sense of \citet{BorgersLiWang2025}.\footnote{A mechanism is dominated if another mechanism yields the principal a weakly higher payoff under every distribution $\dist$ and a strictly higher payoff under some $\dist$.}

\section{Properties and Comparative Analysis}

\subsection{Rationalizable Mechanisms}

We first show that the characterization of the optimal mechanism in \cref{thm:optimal_promise} is essentially the strongest possible without imposing further parametric assumptions on the primitives, by demonstrating that any such promised utility function can arise as the ($F$-almost surely) unique optimal mechanism under some choice of primitives.

Define a candidate promised utility function $u(T)$ to be any function of the form described in \cref{thm:optimal_promise}. Such a function is fully determined by the complete-information curve $\completeInfo(T)$ and at most $K$ departure points from it. Because candidate promised utility functions are far lower-dimensional than the space of model primitives, varying only a subset of primitives suffices to rationalize every candidate promised utility function. To illustrate the high-dimensional nature of the primitives, we consider two classes of variations, each of which can rationalize all candidate promised utility functions.
To simplify the exposition, we will assume that the complete information curve $\completeInfo(\type)$ is continuous in $\type$, and hence the optimal promised utility $\promise(\type)$ is also continuous in $\type$. 

\paragraph{Varying Distribution}

First, we fix the technology set $C(T)$, which determines the complete-information curve $\completeInfo(T)$, and vary the distribution $\dist$ over types $T\in[0,1]$. 
\begin{proposition}\label{prop:implementation_dist}
For any $C(T)$ and any candidate promised utility function $\promise(\type)$, there exists a distribution $\dist$ such that $\promise(\type)$ is $\dist$-a.s.~uniquely optimal.
\end{proposition}

\paragraph{Varying Technology Set}
Second and, more importantly, we fix the complete-information curve $\completeInfo(T)$ and the distribution $\dist$ over types $T\in[0,1]$ and vary the technology set by changing the choice set $C(T)$ associated with each type $T$. Here, the type $T\in[0,1]$ represents the rank of the agent in terms of capability.

Note that the complete-information curve $\completeInfo(\type)$ is the most important object of our analysis in this section. Yet $\completeInfo$ is only a one-dimensional projection of the choice-set correspondence $C(\cdot)$: it records the agent's payoff at the principal's preferred feasible outcome but discards the rest of the geometry of $C(\type)$---in particular, how fast the principal's payoff $\optv(\type,u)$ falls as the promised utility moves away from $\completeInfo(\type)$. The following proposition shows that this discarded information is rich enough that any candidate promised utility consistent with a \emph{fixed} curve $\completeInfo$ can be made uniquely optimal by varying only the choice sets behind it. In other words, holding the picture fixed, there remains enough high-dimensional freedom in the primitives to rationalize every mechanism permitted by \cref{thm:optimal_promise}.
\begin{proposition}\label{prop:implementation_tech}
    For any complete-information curve $\completeInfo(T)$, any fully supported $\dist$, and any candidate promised utility function $\promise(\type)$, there exists a monotone choice set function $C(T)$ that is consistent with $\completeInfo(T)$, such that $\promise(\type)$ is $\dist$-a.s.~uniquely optimal.
\end{proposition}

This exercise highlights that although types $\type$ are fully ranked, the model retains a high-dimensional character. The ranking orders agents by capability, but substantial heterogeneity remains in the internal geometry of choice sets and payoff possibilities. Consequently, many distinct underlying primitives are consistent with the same ranked type structure and complete-information curve.

\subsection{Undominated Mechanisms}
\label{sec:undominated}

\cref{prop:implementation_dist} shows that every mechanism of the form in \cref{thm:optimal_promise} is optimal under some distribution of types. We now establish a stronger robustness property: these are exactly the mechanisms that survive a dominance test that makes no reference to any particular distribution. This is useful because the principal may not know the distribution $\dist$ of agent capabilities, and a mechanism optimal only for a finely-tuned $\dist$ offers little guarantee. Following \citet{BorgersLiWang2025}, we adopt the following criterion.

\begin{definition}[Dominance]
\label{def:undominated}
A feasible and incentive-compatible promised utility $\promise$ is \emph{dominated} by another feasible and incentive-compatible promised utility $\promise'$ if $\promise'$ yields the principal a weakly higher expected payoff under every distribution $\dist$ and a strictly higher expected payoff under some distribution $\dist$. A mechanism is \emph{undominated} if no feasible and incentive-compatible mechanism dominates it.
\end{definition}

\begin{proposition}[Undominated mechanisms]
\label{prop:undominated}
A mechanism is undominated if and only if it takes the form described in \cref{thm:optimal_promise}: the type space partitions into intervals on each of which the promised utility either coincides with the complete-information curve $\completeInfo(\type)$ or is constant.
\end{proposition}

\Cref{prop:undominated} sharpens the sense in which \cref{thm:optimal_promise} identifies the ``right'' class of mechanisms. To understand it, note that
the flattening operation in the proof of \cref{thm:optimal_promise} improves the principal's payoff \emph{pointwise} in the type, $\optv(\type,\hat{\promise}(\type))\geq \optv(\type,\promise(\type))$ for every $\type$. This comparison does not depend on $\dist$ because the objective function is the principal's true value function instead of a virtual value function modified by the distribution $\dist$, as ironing in the classical screening problem would. This distribution-free improvement gives one direction of the result; the rationalization result of \cref{prop:implementation_dist} gives the other.

\begin{proof}
\emph{(Only if.)} Let $\promise$ be feasible, incentive compatible, and not of the form in \cref{thm:optimal_promise}. By \cref{lem:range_of_promised_util}, we may take $\promise$ weakly increasing and within the monotone envelope. The construction in the proof of \cref{thm:optimal_promise} yields a weakly increasing $\hat{\promise}$ of the bang--bang form, in the monotone envelope, with $\optv(\type,\hat{\promise}(\type))\geq \optv(\type,\promise(\type))$ for every $\type$ and strict inequality on an interval of positive length, where $\promise$ lies strictly between $\completeInfo$ and the relevant constant so that \eqref{eq:singlepeakedV} and \cref{assumption_generic} make the gain strict. Since this improvement is pointwise in $\type$, $\hat{\promise}$ raises the principal's expected payoff weakly under every $\dist$ and strictly under any $\dist$ charging that interval. Hence $\promise$ is dominated.

\emph{(If.)} Let $\promise$ be of the form in \cref{thm:optimal_promise}. By \cref{prop:implementation_dist}, there exists a distribution $\dist$ under which $\promise$ is $\dist$-a.s.\ uniquely optimal. If some feasible, incentive-compatible $\promise'$ dominated $\promise$, it would have to yield the principal a weakly higher expected payoff under every distribution, in particular under this $\dist$; but $\promise$ is the unique maximizer at $\dist$, so any $\promise'$ differing from $\promise$ on a positive-measure set gives strictly lower expected payoff at $\dist$, a contradiction. Hence $\promise$ is undominated.
\end{proof}

\subsection{Monotone Comparative Statics}
\label{sec:mcs}

Next, we derive monotone comparative statics for the principal and the agent's optimal payoff.
Throughout, we maintain Assumption~\ref{assumption1} and identify types with $\type\in[0,1]$
so that $\type\le \type'$ if and only if the technology set of type $\type$ is contained in that of type $\type'$.

For a type distribution $\dist$ on $[0,1]$, let $\opt(F)$ denote the value of \eqref{eq_designproblem} and let $\promise^*$ be the corresponding optimal promised utility.

\paragraph{Technology Expansions}

We first consider a comparative static that improves the agent's feasible technologies \emph{type-by-type}.
In many applications, technological progress, improved institutional capacity, or regulatory changes expand the
set of feasible technologies available to each agent. Importantly, such expansions may be heterogeneous across
types and therefore may change the support of the type distribution; moreover, even if the baseline type space is
totally ordered by set inclusion (Assumption~\ref{assumption1}), the \emph{expanded} technology sets need not remain
nested. Nonetheless, the principal's optimal expected payoff is monotone under such expansions.

\begin{definition}[Technology expansion dominance]\label{def:tech_dominance}
We say that $\hat F$ \emph{dominates} $\dist$ in the sense of technology expansion, denoted $\hat F\succeq F$,
if there exists a measurable map $\phi$ such that
\[
\phi(\type)\supseteq \type \quad\text{for all }\type,
\qquad\text{and}\qquad
\hat \type=\phi(\type)\ \text{ when }\ \type\sim F,
\]
so that $\hat F$ is the distribution of $\hat\type$ induced by $\phi$.
\end{definition}

\begin{proposition}[Principal's payoff under technology expansion]\label{prop:dominance}
If $\hat F\succeq F$ and $\dist$ satisfies Assumption~\ref{assumption1}, then $\opt(\hat F)\ge \opt(F)$.
\end{proposition}

Definition~\ref{def:tech_dominance} allows the post-expansion distribution $\hat F$ to violate the set-inclusion order:
$\phi$ may expand technologies in a type-dependent way, so that $\phi(\type)$ and $\phi(\type')$ need not be comparable
even when $\type\subseteq \type'$. Our argument only uses that the \emph{baseline} distribution $\dist$ satisfies
Assumption~\ref{assumption1}.

The proof of Proposition~\ref{prop:dominance} uses the following technical lemma, which shows that one can select
an optimal mechanism (for the baseline problem) whose induced principal payoff is monotone along the baseline chain.

\begin{lemma}[Monotone interim principal payoff]\label{lem:increasing_principal_v}
% Suppose $\dist$ satisfies Assumption~\ref{assumption1}. 
There exists an optimal promised utility $\promise^*$ for
\eqref{eq_designproblem} such that $\type\mapsto \optv(\type,\promise^*(\type))$ is weakly increasing.
\end{lemma}

To prove \cref{prop:dominance}, intuitively, the principal can directly adopt the optimal promised utility $\promise^*$ for the dominated distribution $\dist$,
selected so that the induced principal payoff $V^*(\type)\equiv \optv(\type,\promise^*(\type))$ is weakly increasing in the type of the agent (Lemma~\ref{lem:increasing_principal_v}). 
More specifically, let $\pi(S)$ be the largest type in support of $\dist$ contained in $S$, and the mechanism promises
$\promise^*(\pi(S))$ and then chooses the payoff-maximizing outcome in $C(S)$ delivering that promised utility.
Feasibility holds because $\pi(S)\subseteq S$. Incentive compatibility holds because underreporting shrinks $S$, which
can only (weakly) decrease $\pi(S)$ and hence the promised utility. Finally, at the true expanded set $\hat\type$ we have
$\pi(\hat\type)\ge \type$, so (by monotonicity of $V^*$) the principal's payoff is at least her payoff at $\type$.
Averaging over $\type\sim F$ yields $\opt(\hat F)\ge \opt(F)$.

Technology expansion dominance is a \emph{pointwise} improvement: each realized technology set is (weakly) enlarged.
A different and standard comparison is a \emph{distributional} improvement in the scalar capability index in the sense of
first-order stochastic dominance (FOSD). The next corollary records that the principal's value is also monotone under FOSD.

\begin{corollary}[First-order stochastic dominance]\label{cor:fosd}
Suppose $\hat F$ first-order stochastically dominates $\dist$. Then $\opt(\hat F)\ge \opt(F)$.
\end{corollary}

Under first-order stochastic dominance, there exists a standard monotone coupling $\phi$ such that $\hat\type$ is a technological expansion of distribution $\type$ based on $\phi$. Therefore, $\hat F\succeq F$ in the sense of Definition~\ref{def:tech_dominance}. Proposition~\ref{prop:dominance} then
immediately implies $\opt(\hat F)\ge \opt(F)$, and \cref{cor:fosd} holds.

When the technologies available to the agent expand, their effects on the agent's utility are ambiguous. Intuitively, whether the agent's utility improves depends on the alignment between the principal and the agent. To demonstrate this idea, we provide a simple illustration using binary types for the agent. 

Specifically, consider the example with binary types $\type_0 \leq \type_1$. Let $q$ be the probability of $\type_1$ in the prior distribution $\dist$. Note that when $q$ increases, the distribution increases in the FOSD sense. We also assume that $\utilmin\leq \completeInfo(\type_0)$ and $\utilmin\leq \completeInfo(\type_1)$,
where $\utilmin = \min_{(u,v) \in C(\default)} u$.
We consider the impact of $q$ on the agent's utility in two cases.
\begin{itemize}
\item $\completeInfo(\type_0) < \completeInfo(\type_1)$. 
In this case, the optimal promised utility aligns with the complete information curve. When $q$ increases, the agent receives a higher promised utility with a higher probability. The ex ante payoff of the agent increases. 

\item $\completeInfo(\type_0) > \completeInfo(\type_1)$.
In this case, there is a conflict of interests in the incentives. Due to the incentive compatibility constraint, \cref{thm:optimal_promise} implies that the principal offers a constant promised utility to the agent in the optimal mechanism. Let $u_q$ be this constant utility. 
Note that in this example, $u_q$ is decreasing in~$q$. 
This is because the principal's payoff given any constant promised utility $u\in [\completeInfo(\type_1),\completeInfo(\type_0)]$ is 
\begin{align*}
(1-q)\cdot V(\type_0,u) + q\cdot V(\type_1,u)
\end{align*}
which is concave in $u$. 
Moreover, its derivative is $(1-q)\cdot V_u(\type_0,u) + q\cdot V_u(\type_1,u)$, where $V_u(\type_0,u)$ is decreasing and weakly positive, and $V_u(\type_1,u)$ is decreasing and weakly negative for $u\in [\completeInfo(\type_1),\completeInfo(\type_0)]$. 
Therefore, by increasing $q$, the maximizer $u_q$ decreases. 
The ex ante payoff of the agent decreases. 
\end{itemize}

\paragraph{Expansion of the Default Set of Technologies}
While technological expansion in general may have ambiguous effects on the agent, we show that expanding the \emph{default} set of technologies unambiguously improves the payoff of the principal while admitting an optimal selection of promised utilities that weakly lowers the agent's payoff pointwise.
To study comparative statics with respect to the default set, we allow the default set to be an arbitrary set
$\default\subseteq \techs$ and maintain that
\[
\default \subseteq T \quad \text{for every type } T \text{ in the support.}
\]
Thus, an agent of type $T$ can report any set $S$ such that $\default\subseteq S\subseteq T$, and if he discloses
nothing beyond the default, the principal can still choose a payoff in $C(\default)$.

Let
\begin{equation}\label{eq:utilmin_t0}
\utilmin(\default)\;\equiv\;\min_{(u,v)\in C(\default)} u
\end{equation}
denote the lowest payoff the principal can deliver to the agent using only default technologies. By set inclusion
and monotonicity of $C(\cdot)$, if $\default'\subseteq \default$, then $C(\default')\subseteq C(\default)$ and hence
\[
\utilmin(\default)\;\le\;\utilmin(\default').
\]

\begin{proposition}[Principal's payoff under default expansion]\label{prop:outside_option}
For any $\default'\subseteq \default$, the
principal's optimal expected payoff is weakly higher under the default set of technologies $\default$ than under $\default'$.
\end{proposition}
\begin{proof}
For any $\default'\subseteq \default$, any promised utility $\promise(\type)$ that is optimal under the default set of technologies $\default'$ remains feasible and incentive compatible for the principal under the default set of technologies $\default$, as $\utilmin(\default)\;\le\;\utilmin(\default')$. Therefore, the optimal payoff of the principal is weakly higher under the default set of technologies $\default$. 
\end{proof}

\medskip

The following result illustrates the monotonicity of the optimal promised utility of the agent. Intuitively, a smaller set of defaults increases the agent's utility pointwise for all types in the optimal mechanism in order to incentivize the truthful reporting of the agent, due to the fact that a smaller set of defaults leads to a higher utility for the agent in the worst-case punishment for deviation. 

\begin{proposition}[Agent's payoff under default expansion]\label{prop:promises_default}
Fix $\default'\subseteq \default$. Let $\promise^{\default}(\cdot)$ be any optimal promised utility function in the economy with
default $\default$. Then there exists an optimal promised utility function $\promise^{\default'}(\cdot)$ in the economy with
default $\default'$ such that
\[
\promise^{\default'}(\type)\ \ge\ \promise^{\default}(\type)\qquad \forall \type.
\]
\end{proposition}

\section{Applications}

We now return to the three applications introduced in abstract terms in \cref{subsec:applications}---managing persuasion, action elicitation, and production-technology elicitation---and develop a concrete, fully specified version of each. In every case we use the tools developed for the general model: we first compute the complete-information curve $\completeInfo(\type)$, then apply \cref{thm:optimal_promise} to characterize the optimal promised utility, and finally translate the resulting mechanism into the language of the application. These examples serve several purposes. They show that the abstract primitives---technologies, choice sets, and the nested-set order---arise naturally in economically familiar environments, and they illustrate how the complete-information benchmark, its monotone envelope, and the resulting bang--bang structure combine to deliver sharp and sometimes surprising predictions. For instance, even a completely flat promised-utility curve can generate rich economic predictions, because the application-specific structure of how that promised utility is delivered---the optimization that the general characterization deliberately abstracts away---reemerges as the substantive content of the solution.

\subsection{Managing Persuasion}
Consider a sender-receiver persuasion game \citep{kamenica2011bayesian} with binary states $\theta\in\{0,1\}$ and prior $p \triangleq \Pr(\theta=1)\in(0,1)$.
The receiver chooses whether to approve ($a=1$) or reject ($a=0$). The sender's payoff is state independent:
\[
u(a,\theta)=\mathbf{1}\{a=1\},
\]
so the sender only cares about the approval probability. The receiver receives a payoff of 1 for matching the state and a payoff of 0 otherwise: 
\[
v(a,\theta)=\mathbf{1}\{a=\theta\}.
\]
We assume $p \in (0, \frac{1}{2})$, so the principal prefers rejection at the prior.

A test $\pi_t$, indexed by $t\in[0,\frac12]$, generates a binary signal $\{s_0,s_1\}$ with accuracy $t$:
\[
\Pr[s_1\mid \theta=1]=\Pr[s_0\mid \theta=0]=\frac12+t,
\qquad
\Pr[s_1\mid \theta=0]=\Pr[s_0\mid \theta=1]=\frac12-t.
\]
Thus, larger $t$ corresponds to a more informative binary symmetric test. A sender with type $\type\in[0,\frac12]$ has access to $\pi_\type$ and all tests that are Blackwell less informative than $\pi_\type$. Within the binary-symmetric family, Blackwell dominance is equivalent to the order on accuracy, so in particular the sender with type $\type$ can use every binary test $\pi_t$ with $t\le \type$. Hence, the feasible technology sets are nested.

The default set of technologies is the uninformative test $\pi_0$. Since the receiver can always reject, the sender's payoff from the default report is $\utilmin=0$.

\paragraph{Complete-information benchmark.}
Fix a type $\type$. Under complete information, the receiver chooses the Blackwell most informative feasible experiment $\pi_\type$ and then applies the optimal approval rule.

Let $\mu_1(\type)$ and $\mu_0(\type)$ denote the posterior beliefs after $s_1$ and $s_0$ under $\pi_\type$. By the definition of test $\pi_T$, we have
\[
\mu_1(\type)=\frac{p(\frac12+\type)}{p(\frac12+\type)+(1-p)(\frac12-\type)},
\qquad
\mu_0(\type)=\frac{p(\frac12-\type)}{p(\frac12-\type)+(1-p)(\frac12+\type)}.
\]
The receiver approves if and only if approval gives a weakly higher expected payoff than rejection, i.e., $\mu\geq \frac{1}{2}$.
Since the receiver rejects at the prior and $\mu_0(\type)\le p<\frac12$ for all $\type$, she never approves after signal $s_0$. After signal $s_1$, she approves if and only~if
\[
p\Big(\frac12+\type\Big)\ge (1-p)\Big(\frac12-\type\Big).
\]
Therefore, there exists a unique cutoff
\[
\hat{\type}\equiv \frac12-p\in\Big(0,\frac12\Big)
\]
such that the receiver is indifferent after $s_1$ under $\pi_{\hat{\type}}$, and approves if $T > \hat{T}$.

It follows that if $\type<\hat{\type}$, the receiver rejects after both signals, so the sender's payoff is $0$. If $\type>\hat{\type}$, the receiver approves after $s_1$ and rejects after $s_0$, so the sender's payoff is
\[
\completeInfo(\type)= \Pr[s_1]
= p\Big(\frac12+\type\Big)+(1-p)\Big(\frac12-\type\Big)
= \frac12+(2p-1)\type.
\]
Hence, ignoring the knife-edge type $\type=\hat{\type}$, the complete-information curve is
\[
\completeInfo(\type)=
\begin{cases}
0, & \type<\hat{\type},\\[4pt]
\frac12+(2p-1)\type, & \type>\hat{\type}.
\end{cases}
\]

\paragraph{Optimal mechanism.}
In this example, $\completeInfo(\type)$ jumps up at $\hat{\type}$ and is strictly decreasing on $(\hat{\type},\frac12]$. Since $\completeInfo(\type)$ has exactly one decreasing interval, we have $K=1$. By \cref{thm:optimal_promise}, the optimal promised utility takes the form
\[
\promise^*(\type)=
\begin{cases}
0, & \type<\hat{\type},\\[4pt]
\bar u, & \type\ge \hat{\type},
\end{cases}
\]
for some constant $\bar u\in[0,\completeInfo(\hat{\type})]$.

This direct mechanism admits a simple indirect implementation, which we call a \emph{quota rule with a minimum standard}. The sender is allowed to choose any feasible experiment from his technology set, and the receiver commits ex ante to the following policy with quota $\bar u$ and minimum standard $\frac12$. Note that here the minimum standard $\frac{1}{2}$ is chosen such that the receiver is indifferent between accepting and rejecting. After observing the chosen experiment, the receiver ranks signal realizations by the induced posterior belief $\mu$ and approves greedily from the highest posterior downward, never approving any realization with $\mu<\frac12$, until either the expected approval probability reaches $\bar u$ or no further realization satisfying the minimum standard remains; if necessary, he randomizes at the marginal realization. Thus, the sender's payoff from any chosen experiment equals the smaller of $\bar u$ and the experiment's approval capacity subject to the minimum standard. When the sender is indifferent across experiments, ties are broken in favor of the receiver.

This indirect mechanism implements the same payoff as the direct mechanism with promised utility $\promise^*(\type)$. First, if $\type<\hat{\type}$, then even under the most informative feasible binary test $\pi_\type$, the posterior after signal $s_1$ is below the receiver's minimum standard. Since every feasible experiment is Blackwell less informative than $\pi_\type$, no feasible experiment can generate a posterior above the minimum standard. Therefore, the sender's payoff is $0$, coinciding with $\promise^*(\type)$.

Now consider $\type\ge \hat{\type}$. 
\begin{itemize}
\item If $\completeInfo(\type)\ge \bar u$, then the sender can attain the quota by choosing $\pi_\type$, after which the receiver approves signal $s_1$ with probability $\frac{\bar u}{\completeInfo(\type)}$ 
and rejects all $s_0$ realizations. 
This selection maximizes the receiver's utility subject to the promised utility constraint, since it only approves the highest-posterior realizations under the most informative experiment to meet the quota.

\item If $\completeInfo(\type)<\bar u$, the sender chooses a less informative feasible experiment to increase the probability of signals that satisfy the minimum standard. Since any feasible experiment is Blackwell less informative than $\pi_\type$, the sender can equivalently choose a garbling of the binary test $\pi_\type$. 
Specifically, the sender pools all $s_1$ realizations with just enough $s_0$ realizations to make the probability of the approved message equal to $\bar u$, and sends the remaining $s_0$ realizations to a rejected message. Formally, let
$\alpha(\type)\equiv \frac{\bar u-\completeInfo(\type)}{1-\completeInfo(\type)}\in(0,1]$.
The sender chooses the garbling
\[
m_1=
\begin{cases}
s_1,& \\
s_0,& \text{with probability }\alpha(\type),
\end{cases}
\qquad
m_0=\text{the remaining } s_0 \text{ realizations}.
\]
By definition, we have
\[
\Pr(m_1)=\completeInfo(\type)+\alpha(\type)(1-\completeInfo(\type))=\bar u.
\]
Because $\type\ge \hat{\type}$ and $\bar u\le P(\hat{\type})$, the posterior induced by $m_1$ is weakly above $\frac12$, so $m_1$ satisfies the minimum standard. The receiver therefore approves $m_1$ and rejects $m_0$, and the sender obtains exactly $\bar u$.
Finally, it is easy to verify that this garbling is receiver-optimal among all experiments that promise the sender utility $\bar u$.
\end{itemize}

This example illustrates the central logic of screening for choice sets in a persuasion environment. When $p<\frac12$, a more informative experiment is better for the receiver under complete information, but it is worse for the sender because the probability of a positive signal falls with test accuracy. Hence, the sender's complete-information payoff is decreasing over the informative region. The receiver therefore cannot simply reward more capable senders according to the outcome she would choose under complete information: doing so would give high-capability senders an incentive to conceal precision and mimic lower-capability senders. The optimal mechanism instead flattens the sender's promised approval probability even if the posterior belief exceeds the cutoff of $\frac{1}{2}$. 
Its indirect implementation is a quota rule with a minimum standard: the receiver guarantees a fixed approval probability whenever the sender can generate sufficiently favorable evidence, but approves only signals whose posterior meets the receiver's standard and uses the most receiver-favorable way to deliver the quota. Thus, richer feasible experiment sets need not translate into higher sender payoffs; instead, they allow the receiver to improve how approval is allocated across signal realizations while keeping the sender's incentive to disclose disciplined.

\subsection{Action Elicitation}
Consider the example of civil servants. 
In this application, there is a binary state $\theta \in \{-1,1\}$, and the probability of the state being $1$ is $q\in (0,1)$. The politician receives a private signal $s$ about the state $\theta$ such that the posterior is $\mathbb{P}(\theta=1|s)=s$. Let $G$ be the distribution over signal $s$. 
The politician always has two actions: the status quo $a_0$ and fighting with civil servants $a_f$. 
This is the default set of technologies available to the politician.
The politician's utility $v$ and the civil servant's utility $u$ under these two actions are
\begin{align*}
    &v(a_0,\theta)=0  &u(a_0,\theta)= 0;\\
    &v(a_f,\theta)=-m  &u(a_f,\theta)= -n.
\end{align*}
For the purpose of simplifying the exposition, we assume $n\geq 1$. Besides these two actions, there are potentially many actions for reforms $a_d^i$. A reform $a_d^i$ is indexed by the direction $i\in\{-1,1\}$ and the radical degree $d$. 
The politician prefers to conduct the correct reform but dislikes reforms in the wrong direction:
\begin{align*}
    v(a_d^i,\theta) = d \text{ if }i=\theta, \quad   v(a_d^i,\theta) = -d  \text{ if }i\neq \theta.
\end{align*}
The civil servant, on the other hand, is relatively biased against the politician in that they hate changes so 
\begin{align*}
    u(a_d^i,\theta) = d - \alpha d \text{ if }i=\theta, \quad   u(a_d^i,\theta) = -d  -\alpha d \text{ if }i\neq \theta.
\end{align*}

Only the civil servant privately knows how to conduct reforms, and the civil servant has a private capability type $\type\in[0,1]$ drawn according to distribution $\dist$. A civil servant with type $\type$ knows how to conduct $a_d^i$ for any $i=\{-1,1\}$ and for any $d \in [0,\type]$.

\paragraph{Complete-information benchmark.}
The solution depends crucially on the politician's capability $c_p$, where
\begin{align*}
    c_p = \int_0^1  \max \{2s-1, 1-2s \}  \dd G(s)
\end{align*}
This is the politician's payoff if she can optimally use the most radical reforms according to her own private signals. 
Since $c_p>0$, for every $\type>0$ the politician strictly prefers the most radical feasible reform $a_\type^i$ (directed by her posterior) to the status quo (payoff $0$) and to fighting (payoff $-m$); at $\type=0$, status quo and zero reform are payoff-equivalent.
The agent's payoff at the politician's preferred action is $(c_p-\alpha)\type$, while the default-punishment floor is $\utilmin=\min\{0,-n\}=-n$. Hence the complete information curve is
\begin{align*}
\completeInfo(\type) = \max\{-n,\,(c_p-\alpha)\,\type\}.
\end{align*}
We assume $n \geq \alpha-c_p$, so the floor $-n$ never binds on $\type\in[0,1]$ and $\completeInfo(\type)=(c_p-\alpha)\type$. 
Thus $\completeInfo$ is weakly increasing when $c_p\ge\alpha$ and strictly decreasing when $c_p<\alpha$, so it has at most one decreasing interval.

\paragraph{Optimal Mechanism.}
Applying \cref{thm:optimal_promise}, the optimal promised utility depends on the sign of $c_p-\alpha$.
\begin{enumerate}
\item $c_p \geq \alpha$. In this case, the complete information curve is $\completeInfo(\type)$ increasing in the agent's type by allowing the politician to optimally adopt the most radical reform for all types. 
By \cref{thm:optimal_promise}, the optimal promised utility coincides with the complete information curve. 

\item $c_p < \alpha$. In this case, the complete information curve is $\completeInfo(\type)$ strictly decreasing in the agent's type. 
By \cref{thm:optimal_promise}, the optimal promised utility is a constant function. 
To implement this constant promised utility $\bar{u}$ while maximizing the politician's payoff, 
there exists $\hat{\type}\in [0,1]$ such that 
\begin{itemize}
\item If $\type \leq \hat{\type}$, the politician commits to adopting the reform $a_{\type}^i$ with the direction $i$ that matches her posterior if and only if her posterior is sufficiently precise. 
Otherwise, the politician chooses to fight with the civil servants $a_f$. As $\type$ increases, the probability of $a_f$ decreases.

\item If $\type \geq \hat{\type}$, the politician commits to adopting the reform $a_{\type}^i$ with the direction $i$ that matches her posterior if and only if her posterior is sufficiently precise. 
Otherwise, the politician chooses the status quo $a_0$. As $\type$ increases, the probability of $a_0$ increases.
\end{itemize}
Note that in this case, the mixing probabilities are pinned down by the requirement that the civil servant's expected utility equal $\bar u$ at each type: below $\hat\type$, the reform (worth $(c_p-\alpha)\type>\bar u$ to the agent) is diluted with fighting (worth $-n<\bar u$); above $\hat\type$, it is diluted with the status quo (worth $0>\bar u$, since $\bar u<0$). The reform probability equals $1$ at $\hat\type$ and falls on both sides, so it is maximized for the middle types.
\end{enumerate}

\subsection{Business Strategies with Moral Hazard}
Consider the application of CEO compensation. The firm's output is binary, \(y\in\{0,1\}\), where success generates revenue \(1\) and failure generates revenue \(0\). The board is the principal and the CEO is the agent. The board observes only realized output, not the CEO's implemented strategy. Thus, compensation can depend on output but not directly on whether the CEO exerts effort for a productive business plan. 

The board can commit to an affine wage contract
\[
w(y)=s+by,
\]
where \(s\geq 0\) is a fixed salary and \(b\geq 0\) is a success bonus.

A business strategy is a pair \((G,k)\), where \(G\in\Delta\{0,1\}\) is a distribution over output and \(k\) is the implementation cost borne by the CEO. The firm always has access to a default strategy \((G_0,0)\), interpreted as shirking, where
\[
G_0(y=1)=p_0.
\]
In addition, a CEO with type \(\type\in[\underline{\type},\bar{\type}]\) privately has additional access to all productive business plans with success probabilities at most \(\type\):
\[
\mathcal{G}_B(\type)=\{(G_x,c):x\in[\underline{\type},\type]\},
\qquad G_x(y=1)=x,
\]
where \(c>0\) is the cost of exerting effort and \(p_0<\underline{\type}\). Thus, a higher-type CEO has access to a larger set of feasible business plans, and \cref{assumption1} is satisfied.

After the CEO reports his available business plans, the board can approve a disclosed plan and choose a wage contract to incentivize its implementation. The CEO then privately chooses whether to implement the approved productive plan or to shirk. The CEO can conceal feasible business plans, but he cannot falsify a plan that does not exist, since the success probability and implementation cost of any disclosed plan are verifiable. Hence, the CEO's private information concerns which productive plans are feasible, while moral hazard concerns the effort level he actually implements.

If the CEO reports no productive business plan, the board is left with only the default strategy. By setting \(s=b=0\), the board gives the CEO a payoff \(0\) and obtains a payoff \(p_0\). Hence,
\[
\utilmin=0.
\]

\paragraph{Complete-information benchmark.}
Fix a type \(\type\). Suppose the board wants to induce a disclosed productive business plan \(x\leq \type\). Under contract \(w(y)=s+by\), the CEO's payoff from implementing plan \(x\) is \(s+xb-c\), whereas his payoff from shirking is \(s+p_0b\). Thus, inducing the CEO's effort given plan \(x\) requires
\[
b\geq \frac{c}{x-p_0}.
\]
The board therefore chooses the minimum success bonus $b=\frac{c}{x-p_0}$.
Since a fixed salary only transfers surplus from the board to the CEO and does not affect the incentive constraint, the board sets \(s=0\) when it wants to minimize the CEO's rent.

Under this minimum-incentive contract, the CEO's payoff from plan \(x\) is
\[
\frac{cx}{x-p_0}-c
=
\frac{c p_0}{x-p_0},
\]
and the board's payoff is
\[
x-\frac{cx}{x-p_0}.
\]
The CEO's moral-hazard rent is strictly decreasing in \(x\), while the board's payoff is strictly increasing in $x$. Hence, if the board induces any productive business plan, it is optimal to induce the best available plan \(x=\type\).

The board compares inducing the best available productive plan with using the default strategy. It induces the productive plan if and only if
\[
\type-\frac{c\type}{\type-p_0}\geq p_0.
\]
Since $\type-\frac{c\type}{\type-p_0}$ is strictly increasing in $\type$,
there exists a unique cutoff $\hat{\type}$ where the above inequality holds with equality.
Thus, ignoring the knife-edge type \(\hat{\type}\), if \(\type<\hat{\type}\), the board uses the default strategy and the CEO obtains payoff \(0\). If \(\type>\hat{\type}\), the board induces the best available productive plan \(\type\), and the CEO obtains the moral-hazard rent \(\frac{c p_0}{\type-p_0}\). Therefore, the complete-information curve is
\[
\completeInfo(\type)=
\begin{cases}
0, & \type<\hat{\type},\\[6pt]
\dfrac{c p_0}{\type-p_0}, & \type>\hat{\type}.
\end{cases}
\]
The complete-information curve is flat at zero for low types, jumps up at \(\hat{\type}\), and is strictly decreasing on \((\hat{\type},\bar{\type}]\).

\paragraph{Optimal mechanism.}
In this example, \(\completeInfo(\type)\) has exactly one decreasing interval, namely \((\hat{\type},\bar{\type}]\). Thus, \(K=1\). By \cref{thm:optimal_promise}, the optimal promised utility takes the form
\[
\promise^*(\type)=
\begin{cases}
0, & \type<\hat{\type},\\[4pt]
\bar u, & \type\geq \hat{\type},
\end{cases}
\]
for some constant
\[
\bar u\in
\left[
\completeInfo(\bar{\type}),
\frac{c p_0}{\hat{\type}-p_0}
\right].
\]
If $\type < \hat{\type}$, the optimal promised utility for the CEO is $0$. In this case, the board implements the default plan and provides zero compensation to the CEO. The expected payoff to the board is $p_0$. 

Now consider $\type\ge \hat{\type}$. 
\begin{itemize}
\item If $\promise^*(\type) < \completeInfo(\type)$, approving the business plan \(\type\) under the minimum-incentive contract would give the CEO more than his promised utility. To restore the incentives for truthful reporting, the board has to approve the plan \(\type\) with a probability only
$\frac{\promise^*(\type)}{\completeInfo(\type)} < 1$,
and, conditional on approval, offers the minimum-incentive contract
\[
s=0,
\qquad
b=\frac{c}{\type-p_0}.
\]
With the remaining probability, the board uses the default strategy and offers \(s=b=0\) despite the fact that incentivizing effort under the new business plan is profitable for both parties. 

\item If $\promise^*(\type) \geq \completeInfo(\type)$, the minimum-incentive contract gives the CEO only \(\completeInfo(\type)\), which is below the promised utility. The board therefore pays an additional performance-independent salary
\[
s=\promise^*(\type)-\completeInfo(\type)
=
\promise^*(\type)-\frac{c p_0}{\type-p_0},
\]
while keeping the success bonus at
\[
b=\frac{c}{\type-p_0}.
\]
This raises the CEO's payoff to $\promise^*(\type)$. 
Although the fixed salary is immaterial for incentivizing effort, it is crucial for ensuring that the CEO discloses the best feasible business plan. 
\end{itemize}

This example illustrates the central logic of screening for choice sets in a moral-hazard environment. A more productive business plan is better for the board, but it lowers the CEO's moral-hazard rent because a smaller success bonus is needed to make the productive plan attractive relative to shirking. Hence, the CEO's complete-information payoff is decreasing over the region in which productive plans are induced. The board therefore cannot simply implement the complete-information allocation type by type: high-capability CEOs would prefer to conceal their better plans and mimic lower-capability CEOs in order to obtain larger incentive rents. The optimal mechanism instead flattens the CEO's promised utility. In CEO-compensation terms, this flattening is implemented through two familiar instruments. For intermediate disclosed plans, the board limits approval or investment probability rather than fully exploiting the plan. For highly productive disclosed plans, the board approves the plan for sure, lowers the performance-sensitive bonus, and uses a fixed salary component to maintain the promised compensation level.

\section{Conclusion} \label{sec:conclusion}
This paper studies a screening problem in which private information concerns \emph{what can be done} rather than \emph{how payoffs vary across types}. An agent privately observes which technologies (or proposals, experiments, reforms, business strategies) are feasible and can disclose only a subset to a principal. Because disclosed options are verifiable and their payoff consequences are publicly known, strategic behavior operates through concealment that shrinks the principal's opportunity set, not through information rents in the Myersonian sense. This distinction is central for organizational and regulatory settings in which principals can evaluate proposals once presented but cannot access options that are withheld.

Our model differs substantially from the conventional screening model. First, even under complete information about feasibility, the agent's private access creates a commitment problem: the principal must make withholding unattractive. The optimal benchmark therefore features severe off-path discipline (``shoot-the-agent''), implying that private access primarily pins down a minimum utility that the agent must be guaranteed on-path. Second, once private information over feasible sets is introduced (under the nested-set order), incentive compatibility collapses to a simple monotonicity requirement on the agent's promised utility. The optimal mechanism then has a stark structure: the principal either lets the agent's promised utility track the complete information benchmark, or she deliberately compresses rewards by holding the agent's promised utility locally constant.

This bang--bang structure clarifies when and why institutions exhibit reward compression and the under-use of newly disclosed options. When the complete information benchmark is increasing, the principal can fully exploit additional feasibility while still preserving truthful disclosure; incentives are sharp, and differentiation is fine. When the complete information benchmark locally declines with capability, truthful disclosure requires the principal to flatten promised utilities over a region, which in turn forces her to forgo some myopic exploitation of the richer choice set in that region. Importantly, these distortions are not arbitrary: the number of flat regions is bounded by the number of downward-sloping regions of the complete information benchmark. Economically, this means that the complexity of optimal screening is governed by the extent of conflict revealed by the complete-information benchmark, and optimal ``bunching'' is limited and structured rather than pervasive.

Beyond these specific results, we view the framework itself as the paper's main contribution. Environments in which an agent's central private information concerns the \emph{feasibility} of outcomes are pervasive: medical researchers know which treatments are worth developing, procurement managers know which suppliers are available, and bureaucrats know how to maneuver within organizational rules. Our analysis offers a tractable way to study such settings and a unified rationale for reward compression and the selective under-exploitation of disclosed options. The techniques adapt familiar one-dimensional screening ideas, but the framing opens a range of new economic questions, and we hope it provides a useful foundation for further theoretical and applied work.

\bibliographystyle{apalike}
\bibliography{ref}

\clearpage
% \addcontentsline{toc}{section}{Appendix A}

\appendix
\section{Missing Proofs}
\subsection{Optimal Mechanisms}
\label{apx:proof}
\begin{proof}[Proof of \cref{prop_completeinformation}]
Given any mechanism, the agent's utility is at least $\utilmin$ by not reporting anything and letting the principal choose from the default set of technologies $\default$. 
Among all the choices that ensures the agent a utility at least $\utilmin$, the principal's payoff is at most 
\begin{align*}
\max_{(u,v) \in C(T), u \geq \utilmin }  v.
\end{align*}
Note that the ``shoot-the-agent'' mechanism attains this upper bound while ensuring truth telling. Therefore, it must be the optimal mechanism.
\end{proof}

\begin{proof}[Proof of \cref{cor:agent_payoff_complete}]
If $\principalOpt_1(\type) \geq \utilmin$, the ``shoot-the-agent'' mechanism optimally chooses $\principalOpt(\type)$ when receiving a truthful report of $\type$. 
The agent's utility is $\principalOpt_1(\type)$ in this case. 
If $\principalOpt_1(\type) < \utilmin$, 
since the feasible choice is convex, 
the principal's payoff when promising a utility $z\geq \utilmin\geq\principalOpt_1(\type)$ to the agent is decreasing in $z$. 
Therefore, the principal's payoff is maximized by promising a utility $\utilmin$ to the agent. 
This choice of promised utility is also feasible due to the set inclusion assumption.
Combining the two cases, \cref{cor:agent_payoff_complete} holds. 
\end{proof}

\begin{proof}[Proof of \cref{lemma_PromiseUtility}]
Given any direct mechanism $a(\type)$, consider another mechanism $\hat{a}(\type)$ such that 
$\hat{a}_1(\type) = a_1(\type) = \promise(\type)$, 
and 
\begin{align*}
\hat{a}_2(\type) = \max_v &\quad v\\
\text{s.t.} & \quad (\promise(\type), v) \in C(\type). 
\end{align*}
Since $(\promise(\type), a_2(\type))\in C(\type)$ as mechanism $a(\type)$ must be feasible, 
we have $\hat{a}_2(\type) \geq a_2(\type)$ for any type $\type$. 
Moreover, the construction that $\hat{a}_1(\type) = a_1(\type)$ ensures that mechanism $\hat{a}(\type)$ also incentivize the agent to report $\type$ truthfully.
Therefore, mechanism $\hat{a}(\type)$ is an optimal mechanism that takes the form in the statement of \cref{lemma_PromiseUtility}.
\end{proof}

\begin{proof}[Proof of \cref{lemma_feasible}]
By definition, $\promise(\type)$ is feasible if and only if there exists $v$ such that $(\promise(\type),v)\in C(\type)$. 
Since the projection of $C(\type)$ on its first coordinate is a compact and convex interval, 
letting $\lowerfeasible(\type) = \min_{(u,v)\in C(\type)} u,  \upperfeasible(\type) = \max_{(u,v)\in C(\type)} u$, 
$\promise(\type)$ is feasible if and only if $\promise(\type) \in [\lowerfeasible(\type),\upperfeasible(\type)]$.
\end{proof}

\begin{proof}[Proof of \cref{lemma_Monotone}]
We first prove the only if direction. 
Since the promised utility is incentive compatible, for any type $\type\geq\type'$, 
due to the set inclusion assumption (\cref{assumption1}), 
an agent with type $\type$ can always misreport as $\type'$ to ensure a utility of $\promise(\type')$. Therefore, $\promise(\type)$ must be weakly increasing. 
The condition that $\promise(\default)\geq\utilmin$ is implied by the feasibility. 

Now we prove the if direction. Since $\promise(\type)$ is weakly increasing, and the agent can only misreport as a lower type. Therefore, misreporting weakly lower the agent's utility, and hence the mechanism is incentive compatible. 
\end{proof}

\begin{proof}[Proof of \cref{lem:range_of_promised_util}]
We first prove feasibility of the monotone envelope. For any $\type'\leq \type$, monotonicity of the choice correspondence implies $C(\type')\subseteq C(\type)$, so $\completeInfo(\type')$ is feasible for type $\type$. Hence $\uppercurve(\type)\leq \upperfeasible(\type)$. Also, $\completeInfo(\type')\geq \utilmin$ for every $\type'$, and $\utilmin$ is feasible for every type because the default choice set is available for every report. Hence $\lowercurve(\type)\geq \utilmin\geq \lowerfeasible(\type)$. Therefore any $\promise(\type)\in[\lowercurve(\type),\uppercurve(\type)]$ is feasible by \cref{lemma_feasible}.

We next show that it is without loss to restrict attention to the monotone envelope. Let $\promise$ be any feasible and incentive-compatible promised utility. If $\promise(\hat{\type})>\uppercurve(\hat{\type})$ for some $\hat{\type}$, define
\[
\promise^1(\type)=\min\{\promise(\type),\uppercurve(\type)\}.
\]
The function $\promise^1$ is weakly increasing. Whenever $\promise^1(\type)\neq \promise(\type)$, we have $\promise(\type)>\promise^1(\type)=\uppercurve(\type)\geq \completeInfo(\type)$. Since $u\mapsto\optv(\type,u)$ is concave and is maximized at $\completeInfo(\type)$, lowering the promise down to $\uppercurve(\type)$ weakly increases the principal's payoff. Thus $\promise^1$ weakly improves on $\promise$. Similarly, replacing $\promise^1$ with
\[
\promise^2(\type)=\max\{\promise^1(\type),\lowercurve(\type)\}
\]
preserves monotonicity and feasibility and weakly improves the principal's payoff wherever the lower envelope binds. Hence it is enough to maximize over
\[
\mathcal U
=
\Big\{
\promise:[0,1]\to\reals:
\promise \text{ is weakly increasing and }
\lowercurve(\type)\leq \promise(\type)\leq \uppercurve(\type)
\text{ for every }\type
\Big\}.
\]
Next we show that the maximum exists. 
Note that $\mathcal U$ is nonempty because $\lowercurve\in\mathcal U$. It is also uniformly bounded: since $C(\type)\subseteq C(1)$ for every $\type\in[0,1]$, all feasible promised utilities in the envelope lie in the compact projection of $C(1)$ onto the agent's payoff coordinate.
For any $\promise\in\mathcal U$, the integrand $\type\mapsto\optv(\type,\promise(\type))$ is measurable: for every open set $O\subseteq\reals^2$, the set $\{\type:C(\type)\cap O\neq\emptyset\}$ is an upper interval in $[0,1]$ and hence Borel, and $\promise$ is Borel because it is monotone.

Let $(\promise_n)_{n\geq 1}$ be a maximizing sequence in $\mathcal U$:
\[
\lim_{n\to\infty}
\int \optv(\type,\promise_n(\type))\,\dd F(\type)
=
\sup_{\promise\in\mathcal U}
\int \optv(\type,\promise(\type))\,\dd F(\type).
\]
By Helly's selection theorem, since the functions $\promise_n$ are weakly increasing and uniformly bounded, there is a subsequence, still denoted $(\promise_n)$, and a weakly increasing function $\promise^*$ such that
\[
\promise_n(\type)\to \promise^*(\type)
\qquad\text{for every }\type\in[0,1].
\]
Because the envelope constraints are pointwise closed, $\promise^*\in\mathcal U$.

Fix any type $\type$. Since $C(\type)$ is compact, the value function
\[
u\mapsto \optv(\type,u)=\max\{v:(u,v)\in C(\type)\}
\]
is upper semicontinuous on its feasible domain. Therefore,
\[
\limsup_{n\to\infty}\optv(\type,\promise_n(\type))
\leq
\optv(\type,\promise^*(\type)).
\]
The integrands are uniformly bounded because $C(\type)\subseteq C(1)$ for every $\type$. Hence, by the reverse Fatou lemma,
\[
\limsup_{n\to\infty}
\int \optv(\type,\promise_n(\type))\,\dd F(\type)
\leq
\int
\limsup_{n\to\infty}\optv(\type,\promise_n(\type))\,\dd F(\type)
\leq
\int \optv(\type,\promise^*(\type))\,\dd F(\type).
\]
Thus $\promise^*$ attains the supremum over $\mathcal U$.

Finally, for each report $\type$, choose
\[
a^*(\type)\in
\argmax_{(u,v)\in C(\type)}
\{v:u=\promise^*(\type)\}.
\]
This maximizer exists by compactness. Since $\promise^*$ is weakly increasing, no type can gain by reporting a lower type. Off-path reports can be assigned a payoff in $\argmin_{(u,v)\in C(S)}u$, which gives the agent no more than $\utilmin$. Hence the resulting direct mechanism is feasible, incentive compatible, and optimal.
\end{proof}

\begin{proof}[Proof of \cref{thm:optimal_promise}]
Let $\promise^*$ be an optimal promised utility.
By \cref{lem:range_of_promised_util}, we may assume without loss that $\promise^*$ is weakly increasing and lies in the monotone envelope, hence it is feasible for every $\type$.

Fix any type $\type$. Because $C(\type)$ is convex, the function $u\mapsto \optv(\type,u)$ is concave on its feasible domain.
Moreover, $u_c(\type)=\max\{\utilmin,\principalOpt_1(\type)\}$ is the maximizer of $\optv(\type,u)$ among all feasible promised utilities above $\utilmin$ for type $\type$.
Therefore,
\begin{align}\label{eq:singlepeakedV}
\optv(\type,u)\ \text{is decreasing in $u$ on }[u_c(\type),\infty), 
\text{and increasing in $u$ on }(\utilmin,u_c(\type)].
\end{align}
Moreover, $u_c(\type)$ is unique and the above monotonicity is strict under \cref{assumption_generic}. 

\medskip

\noindent\textbf{Step 1: Bang--bang form (either $u_c$ or constant on each interval).}
Partition the type space into maximal (closed) intervals $\{\hat I_m\}_{m=1}^M$ such that on each $\hat I_m$ exactly one of the following holds:
\[
\promise^*(\type)=u_c(\type)\ \ \forall \type\in \hat I_m;
\qquad
\promise^*(\type)>u_c(\type)\ \ \forall \type\in \hat I_m;
\qquad
\promise^*(\type)<u_c(\type)\ \ \forall \type\in \hat I_m.
\]
(These are maximal in the sense that two adjacent intervals never satisfy the same relation.)

On any interval $\hat I_m$ where $\promise^*(\type)=u_c(\type)$, nothing needs to be changed.

\paragraph{Case A: $\promise^*>u_c$ on $\hat I=[\underline\type,\bar\type]$.}
Define a new promise $\tilde\promise$ on $\hat I$ by
\begin{equation}\label{eq:proj_above}
\tilde\promise(\type)\;\equiv\;
\max\Big\{\promise^*(\underline\type),\ \max_{\underline\type\le s\le \type} u_c(s)\Big\},
\qquad \type\in[\underline\type,\bar\type],
\end{equation}
and set $\tilde\promise(\type)=\promise^*(\type)$ for $\type\notin \hat I$.
Then $\tilde\promise$ is weakly increasing on $\hat I$ (it is a running maximum) and matches $\promise^*$ at the left endpoint.
Also, since $\promise^*$ is weakly increasing and $\promise^*(s)>u_c(s)$ on $\hat I$, we have
\[
u_c(\type)\ \le\ \tilde\promise(\type)\ \le\ \promise^*(\type)
\qquad \forall \type\in\hat I.
\]
By \eqref{eq:singlepeakedV}, lowering the promised utility while staying weakly above $u_c(\type)$ increases the principal's payoff pointwise, hence
\[
\optv(\type,\tilde\promise(\type))\ \ge\ \optv(\type,\promise^*(\type))\qquad \forall \type\in\hat I.
\]
Thus, replacing $\promise^*$ by $\tilde\promise$ on $\hat I$ weakly improves the objective and preserves feasibility and monotonicity.
Finally, by construction of the running maximum on $\hat I$, the function $\tilde\promise$ is either constant or coincides with $u_c$.

\paragraph{Case B: $\promise^*<u_c$ on $\hat I=[\underline\type,\bar\type]$.}
Define a new promise $\tilde\promise$ on $\hat I$ by
\begin{equation}\label{eq:proj_below}
\tilde\promise(\type)\;\equiv\;
\min\Big\{\promise^*(\bar\type),\ \min_{\type\le s\le \bar\type} u_c(s)\Big\},
\qquad \type\in[\underline\type,\bar\type],
\end{equation}
and set $\tilde\promise(\type)=\promise^*(\type)$ for $\type\notin \hat I$.
The map $\type\mapsto \min_{\type\le s\le \bar\type} u_c(s)$ is weakly increasing (a running minimum from the right), hence so is $\tilde\promise$.
Moreover, since $\promise^*(s)<u_c(s)$ on $\hat I$ and $\promise^*$ is weakly increasing,
\[
\promise^*(\type)\ \le\ \tilde\promise(\type)\ \le\ u_c(\type)
\qquad \forall \type\in\hat I.
\]
By \eqref{eq:singlepeakedV}, raising the promised utility while staying weakly below $u_c(\type)$ increases the principal's payoff pointwise, hence
\[
\optv(\type,\tilde\promise(\type))\ \ge\ \optv(\type,\promise^*(\type))\qquad \forall \type\in\hat I.
\]
Thus replacing $\promise^*$ by $\tilde\promise$ on $\hat I$ weakly improves the objective and preserves feasibility and monotonicity.
By construction, on $\hat I$ the function $\tilde\promise$ is either constant or coincides with $u_c$.

\smallskip

Applying the above replacement separately on each strict-inequality interval $\hat I_m$ yields an optimal promised utility (still denoted $\promise^*$)
such that on each maximal interval it is either equal to $u_c(\type)$ or constant.
Let $\{I_j\}_{j\in J_1\cup J_2}$ be the resulting maximal intervals, where $J_1$ indexes those with $\promise^*(\type)=u_c(\type)$
and $J_2$ indexes those on which $\promise^*$ is constant. This proves items (1)--(2) in the theorem.

Moreover, by maximality, between any two distinct $J_1$-intervals there must be at least one $J_2$-interval, so $|J_1|-1\le |J_2|$.

\medskip

\noindent\textbf{Step 2: Bounding the number of constant intervals by $K$.}
Let $\{D_k\}_{k=1}^K$ be the collection of maximal intervals on which $u_c(\cdot)$ is strictly decreasing. 
The fact that $|J_2|\le K$ holds immediately by the following two observations.
\begin{enumerate}
\item \emph{Every constant interval $I_j$ with $j\in J_2$ intersects some $D_k$.}
Let $I_j=[a,b]$ and $\promise(\type)\equiv \bar u$ on $I_j$.
Since $I_j$ is in $J_2$, there exists $\hat{\type}\in(a,b)$ with $\bar u\neq u_c(\hat{\type})$.
Suppose that $u_c$ is weakly increasing. 
If $\bar u<u_c(\hat{\type})$, then $u_c(\type)>\bar u$ for all $\type \in [\hat{T},b]$. 
Increasing the promised utility in the range of $[\hat{T},b]$ by a small $\varepsilon>0$ improves the objective under \cref{assumption_generic} according to \eqref{eq:singlepeakedV}.
Similarly, if $\bar u>u_c(\hat{\type})$, an analogous small decrease in a left-neighborhood improves the principal's payoff.
Therefore, $u_c$ cannot be weakly increasing on $I_j$, so it must strictly decrease somewhere on $I_j$,
and hence $I_j$ intersects with at least one maximal strictly decreasing interval $D_k$.

\item
\emph{Each $D_k$ intersects at most one constant interval $I_j$ with $j\in J_2$.}
Suppose instead that some $D_k$ intersects two distinct constant intervals.
Since $D_k$ is connected and $\promise$ is weakly increasing, there exists $\tau\in D_k$ at which $\promise$ jumps from $\bar u_L$ on the left
to $\bar u_R>\bar u_L$ on the right (both constants on neighborhoods contained in $D_k$).
Because $u_c$ is strictly decreasing on $D_k$, either $\bar u_L<u_c(\tau)$ or $\bar u_R>u_c(\tau)$ (indeed at least one must hold since $\bar u_L<\bar u_R$).
If $\bar u_L<u_c(\tau)$, then for $\type$ just left of $\tau$ we have $\bar u_L<u_c(\type)$, so increasing $\promise$ slightly there (not exceeding $\bar u_R$)
preserves monotonicity and, by \eqref{eq:singlepeakedV}, strictly raises the objective under \cref{assumption_generic}.
If $\bar u_R>u_c(\tau)$, then for $\type$ just right of $\tau$ we have $\bar u_R>u_c(\type)$, so decreasing $\promise$ slightly there (not below $\bar u_L$)
preserves monotonicity and strictly raises the objective.
Either way we contradict optimality. Hence $D_k$ can intersect at most one constant interval.
\end{enumerate}

Combining the bounds gives $|J_1|-1\le |J_2|\le K$, completing the proof.
\end{proof}

\subsection{Properties and Comparative Analysis}
\begin{proof}[Proof of \cref{prop:implementation_dist}]
Given any promised utility $\promise(\type)$, let $\{I_j\}_{j\in J_1}$ be the partitions of intervals such that $\promise(\type)=\completeInfo(\type)$ for any $j\in J_1$ and $\type\in I_j$. 
Let $\dist$ be a distribution with support in $\{I_j\}_{j\in J_1}$. 
By definition, $\promise(\type)$ implements the first best for distribution $\dist$. 
Therefore, $\promise(\type)$ is optimal for distribution $\dist$.
Moreover, under \cref{assumption_generic} where the optimal promised utility is unique for all types of the agent, any deviation of the promised utility for a set of types with positive measure in the support of $\dist$ leads to a strict utility loss for the principal. Therefore, $\promise(\type)$ is $\dist$-a.s.~uniquely optimal.
\end{proof}

\begin{proof}[Proof of \cref{prop:implementation_tech}]
We construct the choice set as 
\begin{align*}
C(T) = \{(u,v): u \in [\min\{\underline{u},\min_{\type} \completeInfo(\type)\}, \max_{\type} \completeInfo(\type)], 
v \in [0, H(T) - \alpha(T)\cdot(u-\completeInfo(\type))^2]\}
\end{align*}
where $\alpha(\type)$ is a non-negative function that is to be determined later, and $H(\type)$ is the parameter chosen to maintain the set inclusion assumption. 

Given any promised utility $\promise(\type)$, let $\{I_j\}_{j\in J_1}$ be the partitions of intervals such that $\promise(\type)=\completeInfo(\type)$ for any $j\in J_1$ and $\type\in I_j$, and  
let $\{I_j\}_{j\in J_2}$ be the partitions of intervals such that $\promise(\type)$ is a constant for any $j\in J_2$ and $\type\in I_j$
Let $\alpha(\type) = 1$ for any type $\type\in I_j$ with $j\in J_1$, and let $\alpha(\type) = 0$ for any type $\type\in I_j$ with $j\in J_2$. 
That is, the optimal utility of the principal is independent of the promised utility of the agent for any type $\type\in I_j$ with $j\in J_2$.\footnote{It is also possible to construct parameters $\alpha(\type) > 0$ for all types such that \cref{assumption_generic} is satisfied, and the candidate $\promise(\type)$ remains $\dist$-a.s.~uniquely optimal when $\promise(\type)$ is an interior solution. The proof of that construction is more involved and hence omitted here to simplify the exposition.} 
Note that in our construction, $\promise(\type)$ implements the first best given the distribution $\dist$. 
Therefore, $\promise(\type)$ is the optimal promised utility. 
Moreover, in our construction, the promised utility for types $\type\in I_j$ with $j\in J_1$ is uniquely pinned down, up to a set with zero measure. 
The continuity of the complete information curve $\completeInfo(\type)$ and the monotonicity of the promised utility function then imply that the optimal promised utilities for types $\type\in I_j$ with $j\in J_2$ are pinned down as well. 
Therefore, $\promise(\type)$ is $\dist$-a.s.~uniquely optimal.
\end{proof}

\begin{proof}[Proof of \cref{lem:increasing_principal_v}]
Let $\promise(\cdot)$ be any feasible and incentive-compatible promised utility.
Define
\[
\bar\type(\type)\in \arg\max_{\type'\le \type}\optv(\type',\promise(\type')),
\]
breaking ties by selecting the \emph{largest} maximizer, and define a new promised utility
\[
\hat\promise(\type)\;\equiv\;\promise(\bar\type(\type)).
\]
Because the maximization set $\{\type'\le \type\}$ expands with $\type$, the tie-breaking rule implies that
$\bar\type(\type)$ is weakly increasing, and since $\promise(\cdot)$ is weakly increasing, $\hat\promise(\cdot)$ is
weakly increasing as well.

Feasibility follows from set inclusion: $\bar\type(\type)\le \type$ implies $C(\bar\type(\type))\subseteq C(\type)$, hence
$[\lowerfeasible(\bar\type(\type)),\upperfeasible(\bar\type(\type))]\subseteq [\lowerfeasible(\type),\upperfeasible(\type)]$.
Since $\promise(\bar\type(\type))$ is feasible for $\bar\type(\type)$, it is also feasible for $\type$.
Therefore $\hat\promise$ is feasible and incentive compatible.

Finally, for every $\type$,
\[
\optv(\type,\hat\promise(\type))
=\optv(\type,\promise(\bar\type(\type)))
\ge \optv(\bar\type(\type),\promise(\bar\type(\type)))
=\max_{\type'\le \type}\optv(\type',\promise(\type')),
\]
where the inequality uses that for fixed $u$, $\optv(\type,u)$ is weakly increasing in $\type$
because $C(\type')\subseteq C(\type)$ whenever $\type'\le \type$.
Thus $\type\mapsto \optv(\type,\hat\promise(\type))$ is the running maximum of
$\type\mapsto \optv(\type,\promise(\type))$ and hence weakly increasing.

Taking expectations under $\dist$ gives
$\int \optv(\type,\hat\promise(\type))\,dF(\type)\ge \int \optv(\type,\promise(\type))\,dF(\type)$.
Applying this construction to an optimal $\promise$ yields an optimal $\promise^*$ with the desired monotonicity.
\end{proof}

\begin{proof}[Proof of Proposition~\ref{prop:dominance}]
Let $\promise^*$ be optimal for the baseline distribution $\dist$ and satisfy Lemma~\ref{lem:increasing_principal_v}.
Define $V^*(\type)\equiv \optv(\type,\promise^*(\type))$, which is weakly increasing in $\type$.

Fix the expansion map $\phi$ in Definition~\ref{def:tech_dominance}. Although the support of $\hat F$ need not be nested,
we will construct a feasible and incentive-compatible mechanism for $\hat F$ that achieves at least $\opt(F)$.

Given any report $S$ (a set of technologies containing the default), define its \emph{projection onto the baseline chain}
as
\[
\pi(S)\;\equiv\;\sup\{\type:\ \type\subseteq S\}.
\]
Because baseline types form a chain under Assumption~\ref{assumption1}, $\pi(\cdot)$ is well-defined and satisfies:
(i) if $S\subseteq S'$, then $\pi(S)\le \pi(S')$; and (ii) $\pi(\type)=\type$ for baseline types.

Now define a promise rule on \emph{all} reports by
\[
\hat\promise(S)\;\equiv\;\promise^*(\pi(S)).
\]
This promise rule is weakly increasing in set inclusion because both $\pi(\cdot)$ and $\promise^*(\cdot)$ are weakly
increasing. Moreover, it is feasible: since $\pi(S)\subseteq S$, any utility feasible for $\pi(S)$ is feasible for $S$
(Lemma~\ref{lemma_feasible}), and $\promise^*(\pi(S))$ is feasible for $\pi(S)$ by construction.

Let the principal, upon report $S$, choose the payoff-maximizing point delivering $\hat\promise(S)$, i.e.,
implement payoff $\optv(S,\hat\promise(S))$. This defines a feasible and incentive-compatible mechanism for $\hat F$.

Now evaluate the principal's payoff at an expanded type $\hat\type=\phi(\type)$. Since $\phi(\type)\supseteq \type$,
we have $\pi(\hat\type)\ge \type$. Therefore,
\begin{align*}
\optv(\hat\type,\hat\promise(\hat\type))
&= \optv\!\big(\hat\type,\promise^*(\pi(\hat\type))\big)
\;\ge\; \optv\!\big(\pi(\hat\type),\promise^*(\pi(\hat\type))\big) \\
&= V^*(\pi(\hat\type))
\;\ge\; V^*(\type),
\end{align*}
where the first inequality uses $\pi(\hat\type)\subseteq \hat\type$ and monotonicity of $\optv(\cdot,u)$ in set inclusion,
and the second uses that $V^*$ is weakly increasing and $\pi(\hat\type)\ge \type$.

Taking expectations with respect to $\type\sim F$ (so that $\hat\type=\phi(\type)\sim \hat F$) yields that the constructed
mechanism attains at least $\int V^*(\type)\,\dd F(\type)=\opt(F)$ under $\hat F$. Hence $\opt(\hat F)\ge \opt(F)$.
\end{proof}

\begin{proof}[Proof of \cref{prop:promises_default}]
Fix any optimal promised utility $\promise^{\default}(\cdot)$  and $\promise^{\default'}(\cdot)$ for the default technologies $\default$ and $\default'$ respectively.
Let $\hat\type$ be the smallest type such that
$\promise^{\default'}(\hat\type) < \promise^{\default}(\hat\type)$, and let $\hat\type_-$ denote the largest type below
$\hat\type$ (so $\promise^{\default'}(\type)\ge \promise^{\default}(\type)$ for all $\type\le \hat\type_-$).\footnote{If $\hat{\type}$ is the lowest type in the support of $\dist$, we define $\hat\type_-$ as the default set of technologies $\default$.}

Consider the subproblem of \eqref{eq_designproblem} restricted to types $\type\ge \hat\type$, taking as given the
boundary condition at $\hat\type_-$: promised utilities must satisfy
\[
\promise(\type)\ \ge\ \promise^{\default}(\hat\type_-)\qquad \forall \type\ge \hat\type,
\]
and must be weakly increasing. Because the objective in \eqref{eq_designproblem} is additively separable across
types and the only cross-type restriction is monotonicity, the restriction of an optimal solution to the tail
$\{\type\ge \hat\type\}$ must solve this continuation problem given the boundary value.

In particular, the restriction of $\promise^{\default}(\cdot)$ to $\{\type\ge \hat\type\}$ is optimal for the tail problem
with the boundary constraint $\promise^{\default}(\hat\type_-)$. Since $\promise^{\default'}(\hat\type_-)\ge
\promise^{\default}(\hat\type_-)$ and $\promise^{\default'}(\cdot)$ is weakly increasing, the tail of $\promise^{\default'}(\cdot)$
is feasible for this tail problem. Hence, the tail of $\promise^{\default}(\cdot)$ yields a weakly higher principal payoff
on $\{\type\ge \hat\type\}$ than the tail of $\promise^{\default'}(\cdot)$.

Now define another promised utility
\[
\widetilde{\promise}(\type)\;=\;
\begin{cases}
\promise^{\default'}(\type), & \type\le \hat\type_-,\\
\promise^{\default}(\type), & \type\ge \hat\type.
\end{cases}
\]
By construction and the definition of $\hat\type$, we have
$\promise^{\default}(\hat\type)\ge \promise^{\default'}(\hat\type)\ge \promise^{\default'}(\hat\type_-)$, so
$\widetilde{\promise}(\cdot)$ is weakly increasing. Moreover, feasibility and the lower-bound constraint are
preserved because $\widetilde{\promise}(\cdot)$ coincides with feasible promised utilities in each region, and
$\default'\subseteq \default$ implies $\utilmin(\default')\ge \utilmin(\default)$.

Therefore, $\widetilde{\promise}(\cdot)$ is feasible and incentive compatible under default $\default'$, and it yields a
weakly higher principal payoff than $\promise^{\default'}(\cdot)$. 
Therefore, $\widetilde{\promise}(\cdot)$ remains optimal when the default set of technologies is $\default'$, and the constraint that 
$\widetilde{\promise}(\type)\ \ge\ \promise^{\default}(\type)$
is satisfied for all type $\type$.
\end{proof}

\end{document}

%% file: macro.tex
\usepackage{amssymb, amsthm, amsmath}
\usepackage{xfrac}
\usepackage{ifthen}
\usepackage{tikz}

\usepackage{enumerate}
\usepackage{hyperref}
\usepackage[capitalize]{cleveref}

\usepackage{subcaption}

\usepackage{amstext}
\usepackage[margin=1.25in]{geometry}
\usepackage[onehalfspacing]{setspace}
\usepackage{natbib}

\usepackage{accents}
\newcommand{\ubar}[1]{\underaccent{\bar}{#1}}

\DeclareMathOperator*{\argmax}{arg\,max}
\DeclareMathOperator*{\argmin}{arg\,min}

\newtheorem{theorem}{Theorem}

\newtheorem{assumption}{Assumption}

\newtheorem{corollary}{Corollary}

\newtheorem{definition}{Definition}

\newtheorem{lemma}{Lemma}

\newtheorem{proposition}{Proposition}

\usepackage{color-edits}
\addauthor{yl}{blue}    % yl for Yingkai
\providecommand{\keywords}[1]
{
  \small	
  \textbf{Keywords---} #1
}

\providecommand{\JEL}[1]
{
  \small	
  \textbf{JEL---} #1
}

%% file: math.tex
\newcommand{\setsize}[1]{{\left|#1\right|}}

\newcommand{\floor}[1]{
{\lfloor {#1} \rfloor}
}
\newcommand{\bigfloor}[1]{
{\left\lfloor {#1} \right\rfloor}
}

%
% probability stuff.
%
\newcommand{\given}{\,\middle|\,}
\newcommand{\wgiven}{\,\mid\,}

% resizing brackets 
\newcommand{\prob}[2][]{\text{\bf Pr}\ifthenelse{\not\equal{}{#1}}{_{#1}}{}\!\left[{\def\givenn{\middle|}#2}\right]}
\newcommand{\expect}[2][]{\text{\bf E}\ifthenelse{\not\equal{}{#1}}{_{#1}}{}\!\left[{\def\givenn{\middle|}#2}\right]}

% brackets do not resize.
\newcommand{\sprob}[2][]{\text{\bf Pr}\ifthenelse{\not\equal{}{#1}}{_{#1}}{}[#2]}
\newcommand{\sexpect}[2][]{\text{\bf E}\ifthenelse{\not\equal{}{#1}}{_{#1}}{}[#2]}

% brackets
\newcommand{\lbr}[1]{\left\{#1\right\}}
\newcommand{\rbr}[1]{\left(#1\right)}
\newcommand{\cbr}[1]{\left[#1\right]}

\newcommand{\suchthat}{\,:\,}

\newcommand{\partialx}[2][]{{\tfrac{\partial #1}{\partial #2}}}
\newcommand{\nicepartialx}[2][]{{\nicefrac{\partial #1}{\partial #2}}}
\newcommand{\dd}{{\,\mathrm d}}
\newcommand{\ddx}[2][]{{\tfrac{\dd #1}{\dd #2}}}
\newcommand{\niceddx}[2][]{{\nicefrac{\dd #1}{\dd #2}}}
\newcommand{\grad}{\nabla}

\newcommand{\symdiff}{\triangle}
\newcommand{\abs}[1]{\left|#1\right|}
\newcommand{\indicate}[1]{{\bf 1}\left[#1\right]}
\newcommand{\reals}{\mathbb{R}}
\newcommand{\posreals}{\reals_+}
\newcommand{\supp}{\text{supp}}

\newcommand{\inftynorm}[1]{\left\lVert#1\right\rVert_{\infty}}

\newcommand{\conv}{{\rm conv}}

%% file: notation.tex
\newcommand{\tech}{t}
\newcommand{\techs}{\mathcal{T}}
\newcommand{\util}{u}
\newcommand{\val}{v}
\newcommand{\type}{T}
\newcommand{\types}{\mathcal{R}}

\newcommand{\outcomes}{\mathcal{Y}}
\newcommand{\maxOutcome}{\bar{y}}

\newcommand{\alloc}{\alpha}
\newcommand{\posterior}{\mu}
\newcommand{\allPosterior}{\Delta \states}
\newcommand{\regret}{R}
\newcommand{\generalRegret}{\regret_{\gamma}}
\newcommand{\dist}{F}
\newcommand{\prior}{\rho}
\newcommand{\state}{\theta}
\newcommand{\states}{\Theta}

\newcommand{\primed}{^\dagger}
\newcommand{\dprimed}{^\ddagger}

\newcommand{\allInfo}{\Sigma}
\newcommand{\info}{\pi}
\newcommand{\actions}{\mathcal{A}}
\newcommand{\action}{a}
\newcommand{\quota}{q}
\newcommand{\availableInfo}{\Pi}
\newcommand{\noInfo}{\info^N}

\newcommand{\indiff}{m^*}
\newcommand{\priorMean}{m_{\prior}}

\newcommand{\feasibleStrategy}{\Psi}
\newcommand{\optV}{V_P^*}
\newcommand{\firstBest}{U_{\rm F}}

\renewcommand{\epsilon}{\varepsilon}
\newcommand{\completeInfo}{u_c}
\newcommand{\uppercurve}{\bar{u}_c}
\newcommand{\lowercurve}{\underline{u}_c}
\newcommand{\promise}{U}
\newcommand{\optv}{V}
\newcommand{\opt}{{\rm OPT}}
\newcommand{\utilmin}{\ubar{u}}
\newcommand{\upperfeasible}{\bar{u}_f}
\newcommand{\lowerfeasible}{\underline{u}_f}
\newcommand{\default}{T_0}
\newcommand{\principalOpt}{a^p}